\documentclass{llncs}

\usepackage{stmaryrd} 
\usepackage{proof}
\usepackage{amssymb}
\usepackage{amsmath}
\usepackage{mathtools}
\usepackage{mathpartir}
\usepackage{xcolor}
\usepackage{graphicx}
\usepackage{url}
\usepackage{ntabbing}

\usepackage{microtype} 
\usepackage[T1]{fontenc} 
\usepackage{lmodern}

\newcommand{\langname}{Resource-Aware SILL}

\newcommand{\W}{\Omega}
\newcommand{\Sg}{\Sigma}

\newcommand{\xvdash}[1]{%
  \vdash^{\mkern-10mu\scriptstyle\rule[-.9ex]{0pt}{0pt}#1}%
}

\newcommand{\potconf}[1]{\overset{#1}{\vDash}}

\newcommand{\measure}{weight}

\newcommand{\semi}{\; ; \;}
\newcommand{\ichoiceop}{\oplus}
\newcommand{\echoiceop}{\with}
\newcommand{\ichoice}[1]{\ichoiceop \{ #1 \}}
\newcommand{\echoice}[1]{\echoiceop \{ #1 \}}

\newcommand{\proc}[2]{\m{proc}(#1, #2)}
\newcommand{\msg}[2]{\m{msg}(#1, #2)}

\newcommand{\ecase}[3]{\m{case} \; #1 \; (#2 \Rightarrow #3)}

\newcommand{\erecvch}[2]{#2 \leftarrow \m{recv} \; #1}

\newcommand{\esendch}[2]{\m{send} \; #1 \; #2}

\newcommand{\ewait}[1]{\m{wait} \; #1}

\newcommand{\eclose}[1]{\m{close} \; #1}

\newcommand{\fwd}[2]{#1 \leftarrow #2}
\newcommand{\m}[1]{\mathsf{#1}}
\newcommand{\esendl}[2]{#1.#2}

\newcommand{\ecut}[4]{#1 \leftarrow #2 \leftarrow \overline{#3} \semi #4}
\newcommand{\espawn}[4]{#1 \leftarrow #2 \leftarrow #3 = #4}

\newcommand{\procdef}[3]{#3 \leftarrow #1 \leftarrow #2}
\newcommand{\procdefna}[2]{#2 \leftarrow #1}
\newcommand{\casedef}[1]{\m{case} \; #1}
\newcommand{\labdef}[1]{#1 \Rightarrow}

\newcommand{\lolli}{\multimap}
\newcommand{\tensor}{\otimes}
\newcommand{\with}{\mathbin{\binampersand}}

\newcommand{\one}{\mathbf{1}}

\newcommand{\mi}[1]{\mbox{\it #1}}

\newcommand{\pot}[2]{#1^{#2}}

\newcommand{\lollipot}[1]{\overset{#1}{\lolli}}
\newcommand{\tensorpot}[1]{\overset{#1}{\tensor}}
\newcommand{\potfop}{\phi}
\newcommand{\potf}[1]{\potfop(#1)}
\newcommand{\mlab}{M^{\textsf{label}}}
\newcommand{\mchan}{M^{\textsf{channel}}}
\newcommand{\mcl}{M^{\textsf{close}}}
\newcommand{\entailpot}[1]{\xvdash{#1}}

\newcommand{\fpot}{\; @ \;}

\newcommand{\entailpotcf}[1]{\underset{\m{cf}}{\entailpot{#1}}}

\newcommand{\config}{\mathcal{C}}

\newcommand{\dc}{\mathcal{D}}
\newcommand{\st}[1]{\m{store}_{#1}}
\newcommand{\stack}[1]{\m{stack}_{#1}}
\newcommand{\queue}[1]{\m{queue}_{#1}}
\newcommand{\mapper}[1]{\m{mapper}_{#1}}
\newcommand{\fdr}[1]{\m{folder}_{#1}}
\newcommand{\lt}[1]{\m{list}_{#1}}
\newcommand{\bits}{\m{bits}}

\newcommand{\step}{\mapsto}

\newcommand{\fresh}[1]{(#1 \text{ fresh})}

\newcommand{\sinfer}[3]
{\inferrule
{#3}
{#2}
\; #1}

\newcommand{\nil}{[]}

\newcommand{\ceil}[1]{\left\lceil #1 \right\rceil}

\newcommand{\comment}[1]{}

\newcommand{\bin}[1]{(#1)_2}
\newcommand{\ctr}{\m{ctr}}

\newcommand{\ignore}[1]{\textcolor{red}{#1}}

\title{Work Analysis with \\ Resource-Aware Session Types
 \vspace{-1em}}
\author{Ankush Das, Jan Hoffmann and Frank Pfenning
 \vspace{-5pt}}
\institute{Carnegie Mellon University, Pittsburgh, PA, USA}

\begin{document}
\maketitle

\vspace{-5ex}

\begin{abstract}
  While there exist several successful techniques for supporting
  programmers in deriving static resource bounds for sequential code,
  analyzing the resource usage of message-passing concurrent processes
  poses additional challenges.
  To meet these challenges, this article presents an analysis for
  statically deriving worst-case bounds on the total work performed by
  message-passing processes.
  To decompose interacting processes into components that can be
  analyzed in isolation, the analysis is based on novel resource-aware
  session types, which describe protocols and resource contracts for
  inter-process communication. A key innovation is that both messages
  and processes carry potential to share and amortize cost while
  communicating.
  To symbolically express resource usage in a setting without static
  data structures and intrinsic sizes, resource contracts describe
  bounds that are functions of interactions between processes.
  Resource-aware session types combine standard binary session types
  and type-based amortized resource analysis in a linear type system.
  This type system is formulated for a core session-type calculus
  of the language SILL and proved sound with respect to a
  multiset-based operational cost semantics that tracks the total
  number of messages that are exchanged in a system.
  The effectiveness of the analysis is demonstrated by analyzing
  standard examples from amortized analysis and the literature on
  session types and by a comparative performance analysis of
  different concurrent programs implementing the same interface.
\end{abstract}

\vspace{-2.5em}

\section{Introduction}

In the past years, there has been increasing interest in supporting
developers to statically reason about the resource usage of their
code. There are different approaches to the problem that are based on
type systems~\cite{Jost03,LagoP13,CicekGA15,CicekBGGH16,HoffmannW15}, abstract
interpretation~\cite{GulwaniMC09,AliasDFG10,CernyHKRZ15}, recurrence
relations~\cite{FloresH14,AlbertGM12,KincaidBBR2017}, termination
analysis~\cite{Zuleger11,BrockschmidtEFFG14,AvanziniLM15,HofmannM15}, and other
techniques~\cite{CarbonneauxHRS17,DannerLR15}.
Among the applications of this research we find the prevention of side
channels that leak secret
information~\cite{NgoDFH16,Antonopoulos17,ChenFD17}, identification of
complexity bugs~\cite{OlivoDL15}, support of scheduling
decisions~\cite{AcarCR16}, and help in
profiling~\cite{HaemmerleLLKGH16}.

While there has been great progress in analyzing sequential code,
relatively little research has been done on analyzing the resource
consumption of concurrent and distributed
programs~\cite{GimenezM16,AlbertFGM16,AlbertACGGMPR15}. The lack of analysis tools is in sharp
contrast to the need of programming language support in this area:
concurrent and distributed programs programs are both increasingly
pervasive and particularly difficult to analyze.  For shared memory
concurrency, we need to precisely predict scheduling to account
for synchronization cost. Similarly, the interactive nature of
message-passing systems makes it difficult to decompose the
system into components that can be analyzed in isolation. After all,
the resource usage of each component crucially depends on its
interactions with the world.

In this paper, we study the foundations of worst-case resource
analysis for message-passing processes. A key idea of our approach is
to rely on \emph{resource-aware session types} to describe structure,
protocols, and resource bounds for inter-process communication that we
can use to perform a compositional and precise amortized analysis.
\emph{Session
types}~\cite{Honda93CONCUR,Honda98esop,Caires10concur,Caires16mscs,Wadler12icfp}
prescribe bidirectional communication protocols for message-passing
processes.  \emph{Binary session types} govern the interaction of two
processes along a single channel, prescribing complementary send and
receive actions for the processes at the two endpoints of a channel.
We use such protocols as the basis of resource usage contracts that
not only specify the type but also the potential of a message that is
sent along a channel.  The potential (in the sense of classic
amortized analysis~\cite{tarjan85}) may be spent sending other
messages as part of the network of interacting processes, or
maintained locally for future interactions.  Resource analysis is
static, using the type system, and the runtime behavior of programs is
not affected.


We focus on bounds on the total work that is performed
by a system, counting the number of messages that are exchanged.
While this alone does not account yet for the concurrent nature of
message-passing programs it constitutes a significant and necessary
first step.
The bounds we derive are also useful in their own right.
For example, the information can be used in scheduling decisions, to
bound the number of messages that are sent along a specific channel,
or to statically decide whether we should spawn a new thread of
control or execute sequentially when possible.
Additionally, bounds on the work of a process can also serve as input
to a Brent-style theorem~\cite{brent2013algorithms} that relates the
complexity of the execution of a program on a $k$-processor machine to
the program's work (the focus of this paper) and span (the resource
usage if we assume an unlimited number of processors). We are working
on a companion paper for deriving bounds on the span, which is both
conceptually and technically quite different.


Our analysis is based on a linear type system that combines standard
binary session types as available in the SILL
language~\cite{Toninho13esop,Pfenning15FOSSACS}, and type-based
amortized resource analysis~\cite{Jost03,HoffmannW15}. Both techniques
are based on linear or affine type systems, making their combination
natural. Each session type constructor is decorated with a natural
number that declares a potential that must be transferred
(conceptually!) along with the corresponding message.
Since the interface to a process is characterized entirely by the
resource-aware session types of the channels it can interact with,
this design provides for a compositional resource specification.  For
closed programs (which evolve into a closed network of interacting
processes) the bound then becomes a single constant.
In addition to the natural compositionality of type systems we also
preserve the good support for deriving resource annotations via LP
solving which is a key feature of type-based amortized analysis. While we
have not yet implemented a type inference algorithm, we designed the
system with support for type inference in mind.
Moreover, resource-aware session types integrate well with type-based
amortized analysis for functional programs because they are based on
compatible logical foundations (intuitionistic linear logic and
intuitionistic logic, respectively), as exemplified in the design of
SILL~\cite{Toninho13esop,Pfenning15FOSSACS} that combines them
monadically.

A conceptual challenge of the work is to express symbolic bounds in a
setting without static data structures and intrinsic sizes. Our
innovation is that resource-aware session types describe bounds that
are functions of interactions (i.e., messages sent) along a channel.
A major technical challenge is to account for the global number of
messages sent with local derivation rules: Operationally, local
message counts are forwarded to a parent process when a sub-process
terminates. As a result, local message counts are incremented by
sub-processes in a rather non-local fashion. Our solution is that both
messages and processes carry potential to share and amortize the cost of a
terminating sub-process proactively as a side-effect of the
communication.

Our main contributions are as follows. We present the first session
type system for deriving parametric bounds on the resource usage of
message-passing processes. We prove the nontrivial soundness of type
system with respect to an operational cost semantics that tracks the
total number of messages exchanged in a network of communicating
process.  We demonstrate the effectiveness of the technique by
deriving tight bounds for some standard examples from amortized
analysis and the literature on session types. We also show how
resource-aware session types can be used to specify and compare the
performance characteristics of different implementations of the same
protocol.
The analysis is currently manual, with automation left for future work, as
is a companion type system for deriving the \emph{span} of concurrent
computations in the same language.


\vspace{-5pt}

\section{Overview}\label{sec:overview}

In this section, we motivate and informally introduce our
resource-aware session types and show how they can be used to analyze
the resource usage of message-passing processes. We start with
building some intuition about session types.

\paragraph{Session Types.}\label{subsec:intro}

Session types have been introduced by Honda~\cite{Honda93CONCUR} to
describe the structure of communication just like standard data types
describe the structure of data. We follow the approach and syntax of
SILL~\cite{Toninho13esop,Pfenning15FOSSACS} which is based on a
Curry-Howard isomorphism between intuitionistic linear logic and
session types, extended by recursively defined types and processes.
In the intuitionistic approach, every channel has a \emph{provider}
and a \emph{client}. We view the session type as describing the
communication from the provider's point of view, with the client
having to perform matching actions.

As a first simple example, we consider natural numbers in binary form.
A process \emph{providing} a natural number sends a stream of bits
starting with the least significant bit.  These bits are represented
by messages $\m{zero}$ and $\m{one}$, eventually terminated by
$\m{\$}$. Because the provider chooses which messages to send,
we call this an \emph{internal choice}, which is written as
$$\m{bits} = \ichoice{\m{zero} : \m{bits}, \m{one} : \m{bits}, \m{\$} : \one} \; .$$
Here, $\ichoice{l_1 : A_1, \ldots, l_n : A_n}$ is an n-ary, labelled
generalization of $A \oplus B$ of linear logic, and $\one$ is the
multiplicative unit of linear logic.  Operationally, $\one$ means the
provider has to send an $\mi{end}$ message, closing the channel and
terminating the communication session.  For example, the number
$6 = (110)_2$ would be represented by the sequence of messages
$\m{zero}, \m{one}, \m{one}, \m{\$}, \mi{end}$.

The session type does not prescribe any particular implementation only
the interface to a process.  In this example, a client of a channel
$c : \m{bits}$ has to branch on whether it receives $\m{zero}$,
$\m{one}$, or $\m{\$}$. Note that as we proceed in a session, the type
of a channel must change according to the protocol.  For example, if a
client receives the message $\m{\$}$ along $c : \m{bits}$ then $c$
must afterwards have type $\one$. The next message along $c$ must
be $\mi{end}$ and we have to wait for that after receiving $\$$ so the
session is properly closed.

As a second example we describe the interface to a counter.  As a
client, we can repeatedly send $\m{inc}$ messages to a counter, until we want to
read its value and send $\m{val}$. At that point the counter will
send a stream of bits representing its value as prescribed by the type
$\m{bits}$. From the provider's point of view, a counter presents an
\emph{external choice}, since the client chooses between $\m{inc}$ or
$\m{val}$.
  $$\m{ctr} = \echoice{\m{inc} : \m{ctr}, \m{val} : \m{bits}}$$
The type former $\echoice{l_1 : A_1, \ldots, l_n : A_n}$ is an n-ary labelled
generalization of $A \echoiceop B$ of linear logic. Operationally, the
provider must branch based on which of the labels $l_i$ it receives.
After receiving $l_k$ along a channel
$c : \echoice{l_1 : A_1, \ldots, l_n : A_n}$, communication along $c$
proceeds at type $A_k$.

Such type formers can be arbitrarily nested to allow more complex
bidirectional protocols. Consider for example the store protocol,
which is defined by the following type.
\begin{center}
\begin{minipage}{3cm}
\begin{tabbing}
$\st{A} = \echoice{$ \= $\m{ins} : A \lolli \st{A},$ \\
\> $\m{del} : \ichoice{\m{none} : \one, \m{some} : A \tensor \st{A}}}$ 
\end{tabbing}
\end{minipage}
\end{center}
A provider of a channel $c$ of type $\st{A}$ either accepts an
$\m{ins}$ or $\m{del}$ message. If it receives the $\m{ins}$ message,
the type of $c$ is now $A \lolli \st{A}$, the implication of linear
logic.  It means the provider now receives a \emph{channel} of type
$A$ along $c$ and then behaves again like a store.  If it receives a
$\m{del}$ message, it responds with either the message $\m{none}$ or
the message $\m{some}$ (an internal choice, $\oplus$). If it sends
$\m{none}$ it next must send an $\mi{end}$ message and terminate.  If
it sends $\m{some}$, the type of $c$ is now $A \tensor \st{A}$.  This
corresponds to the multiplicative conjunction of linear logic and,
operationally, requires the provider to now send a channel of type $A$
and then behave again as $\st{A}$. Linearity of the type system
guarantees that the channels retrieved from a store of this type are
some permutation of the channels inserted.  It may, for example,
behave as a stack or a queue (as explained in
Section~\ref{sec:case_study}).

\paragraph{Modeling a binary counter.}

We will now describe an implementation of a counter and use our
resource-aware session types to analyze its resource usage. Like in
the rest of the paper, the resource we are interested in is the total
number of messages sent along all channels in the system.


A well-known example of amortized analysis counts the number of bits
that must be flipped to increment a counter.  It turns out the
amortized cost per increment is 2, so that $n$
increments require at most $2n$
bits to be flipped. We can see this by introducing a
potential of $1$
for every bit that is $1$
and using the potential to \emph{pay} for the expensive case in which
and increment triggers many flips.  When the lowest bit is zero, we
flip it to one (costing 1) and also store a remaining potential of 1
with this bit. When the lowest bit is one we use the stored potential
to flip the bit back to zero (with no stored potential) and use the
remaining potential of 2 for incrementing the higher bits.

\begin{figure}[t]
\centering
\includegraphics[width=0.6\linewidth]{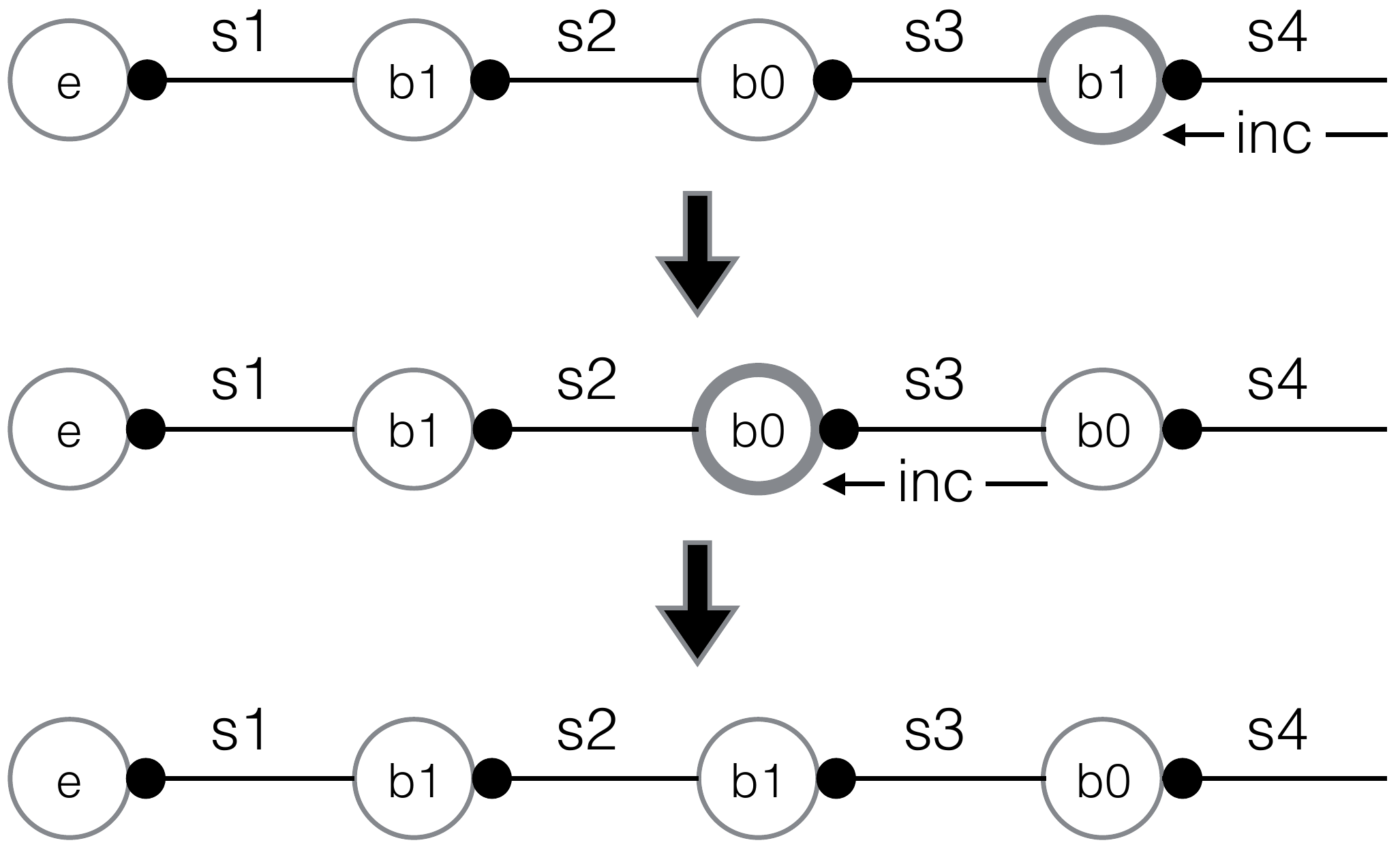}
\caption{A binary counter system representing $5 = \bin{101}$
with the messages triggered when an $\m{inc}$ message
is received along $s_4$.\vspace{-1.5em}}
\label{fig:binary-counter-inc}
\end{figure}
We model a binary counter as a chain of processes where each process
represents a single bit (process $b0$ or $b1$) with a final process
$e$ at the end.
Each of the processes in the chain \emph{provides} a channel
of the $\m{ctr}$ type, and each (except the last) also \emph{uses} a
channel of this type representing the higher bits. For example, in the
first chain in Figure~\ref{fig:binary-counter-inc}, the process $b0$
offers along channel $s_3$ and uses channel $s_2$.  In our notation,
we would write this as
\[
\begin{array}{rcl}
\cdot & \vdash & e :: (s_1 : \m{ctr}) \\
s_1 : \m{ctr} & \vdash & b1 :: (s_2 : \m{ctr}) \\
s_2 : \m{ctr} & \vdash & b0 :: (s_3 : \m{ctr}) \\
s_3 : \m{ctr} & \vdash & b1 :: (s_4 : \m{ctr})
\end{array}
\]
We see that, logically, parallel composition with a private shared
channel corresponds to an application of the cut rule. We do not show
here the \emph{definitions} of $e$, $b0$, and $b1$, which can be found
in Figure~\ref{fig:bincount_impl}.  The only channel visible to an outside client
(not shown) is $s_4$.  Figure~\ref{fig:binary-counter-inc} shows
the messages triggered if an increment message is received along
$s_4$.

\paragraph{Expressing resource bounds.}

Our basic approach is that \emph{messages carry potential} and
\emph{processes store potential}.  This means the sender has to pay
not just 1 unit for sending the message, but whatever additional units
to amortize future costs.  In the amortized analysis of the counter
each bit flip corresponds exactly to an $\m{inc}$ message, because
that is what triggers a bit to be flipped. Our cost model focuses on
messages as prescribed by the session type and does not count other
operations, such as spawning a new process or terminating a
process. This choice is not essential to our approach, but convenient
here.

To capture the informal analysis we need to express \emph{in the type}
that we have to send $1$
unit of potential with the label $\m{inc}$.
We do this using a superscript indicating the required potential with
the label, postponing for now the discussion of $\m{val}$.
$$\m{ctr} = \echoice{\m{inc}^1 : \m{ctr}, \m{val}^? : \m{bits}}\;.$$
When we assign types to the processes, we now use these more
expressive types.  We also indicate the potential stored in
a particular process as a superscript on the turnstile.
\begin{eqnarray}
t : \m{ctr} & \entailpot{0} & b0 :: (s : \m{ctr}) \label{b0_type} \\
t : \m{ctr} & \entailpot{1} & b1 :: (s : \m{ctr}) \label{b1_type} \\
\cdot & \entailpot{0} & e :: (s : \m{ctr}) \label{e_type} 
\end{eqnarray}
With our formal typing rules (see Section~\ref{sec:typing}) we can 
verify these typing constraints, using the definitions of $b0$, $b1$,
and $e$.  Informally, we can reason as follows:
\begin{description}
\item[$b0$:] When $b0$ receives $\m{inc}$ it receives 1 unit of
  potential. It continues as $b1$ (which requires no communication)
  which stores this 1 unit (as prescribed from the type of $b1$ in
  Equation~\ref{b1_type}).
\item[$b1$:] When $b1$ receives $\m{inc}$ it receives 1 unit of
  potential which, when combined with the stored one, makes 2 units.
  It needs to send an $\m{inc}$ messages which consumes these 2
  units (1 to send the message, and 1 to send along a potential of
  1 as prescribed in the type).  It has no remaining potential, which
  is sufficient because it transitions to $b0$ which stores no potential
  (inferred from the type of $b0$ in Equation~\ref{b0_type}).
\item[$e$:] When $e$ receives $\m{inc}$ it receives 1 unit of
  potential.  It spawns a new process $e$ and continues as $b1$.
  Spawning a process is free, and $e$ requires no potential, so it can
  store the potential it received with $b1$ as required.
\end{description}

\begin{figure}[t]
\centering
\includegraphics[width=0.6\linewidth]{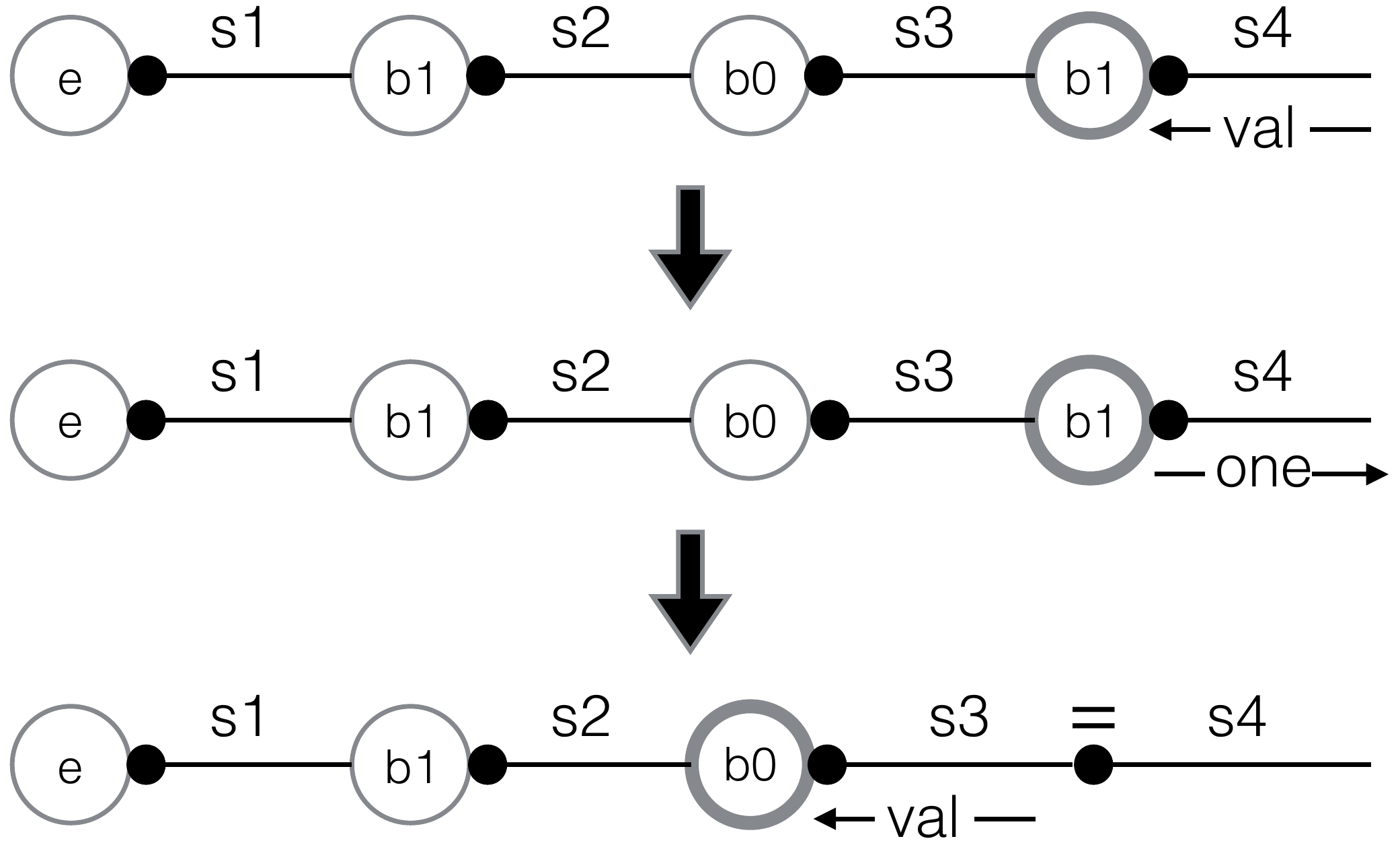}
\caption{A binary counter system representing $5 = \bin{101}$
with the messages triggered when a $\m{val}$ message
is received along $s_4$.\vspace{-1.5em}}
\label{fig:binary-counter-val}
\end{figure}

\noindent How do we handle the type annotation $\m{val}^? : \m{bits}$ of the label $\m{val}$?
Recall that
$\m{bits} = \ichoice{\m{zero} : \m{bits}, \m{one} : \m{bits}, \m{\$} : \one}$.
In our implementation, upon receiving a $\m{val}$ message, a $b0$ or
$b1$ process will first respond with $\m{zero}$ or $\m{one}$.  It then
sends $\m{val}$ along the channel it uses (representing the higher
bits of the number) and terminates by \emph{forwarding} further
communication to the higher bits in the chain.  The $e$ process will
just send $\m{\$}$ and $\mi{end}$, indicating the empty stream of
bits.

We know we will have enough potential to carry out the required send
operations if each process ($b0$, $b1$, and $e$) carries an additional
2 units of potential.  We could impart these with the $\m{inc}$
and $\m{val}$ messages by sending 2 more units with $\m{inc}$ and
2 units with $\m{val}$.  That is, the following type is correct:
\begin{eqnarray*}
\m{bits} &=& \ichoice{\m{zero}^0 : \m{bits}, \m{one}^0 : \m{bits}, \m{\$}^0 : \one^0} \\
\m{ctr} &=& \echoice{\m{inc}^3 : \m{ctr}, \m{val}^2 : \m{bits}}
\end{eqnarray*}
Here, the superscript $0$ in the type of $\m{bits}$ indicates
that the corresponding messages carry no potential.

However, this type is a gross over-approximation.  The processes
of a counter of value $n$,
would carry $2n$
additional units of potential while only $2\ceil{\log(n+1)} + 2$
are needed.  To obtain this more precise bound, we need
\emph{families of session types}.

\paragraph{A more precise analysis.}

A more precise session type for this example requires that \emph{in
  the type} we can refer either to the number of bits in the
representation of a number or its value.  This form of internal
measure is needed only for type-checking purposes, not at runtime.  It
is also not intrinsically tied to a property of a representation, the
way the length of a list in a functional language is tied to its
memory requirements.  We indicate these measures using square
brackets, so that $\m{ctr}[n]$ is a family of types, and
$\m{ctr}[0]$, for example, is a counter with value 0.  Such type
refinements have been considered in the literature on session types
(see, for example, \cite{Griffith13nasa}) with respect to
type-checking and inference.  Here, we treat it as a meta-level
notation to denote families of types.  Following the reasoning
above, we obtain the following type:
\begin{eqnarray*}
\m{bits} &=& \ichoice{\m{zero}^0 : \m{bits}, \m{one}^0 : \m{bits}, \m{\$}^0 : \one^0} \\
\m{ctr}[n] &=& \echoice{\m{inc}^1 : \m{ctr}[n+1], \m{val}^{2\ceil{\log(n+1)}+2} : \m{bits}}
\end{eqnarray*}
To check the types of our implementation, we need to revisit
and refine the typing of the $b0$, $b1$ and $e$ processes.
\[
\begin{array}{lcl}
  t : \m{ctr}[n] & \entailpot{0} & b0 :: (s : \m{ctr}[2n]) \\
  t : \m{ctr}[n] & \entailpot{1} & b1 :: (s : \m{ctr}[2n+1]) \\
  & \entailpot{0} & e :: (s : \m{ctr}[0])
\end{array}
\]
Our type system verifies these types against the
implementation of $b0$, $b1$, and $e$ (see Section~\ref{sec:language}).
The typing rules reduce the well-typedness of these processes
to arithmetic inequalities which we can solve by hand, for
example, using that $\log(2n) = \log(n)+1$.
The intrinsic measure $n$ and
the particular potential annotations are not automatically derived but
come from our insight about the nature of the algorithms and programs.

Before introducing the formalism in which the programs are expressed,
together with the typing rules that let us perform rigorous amortized
analysis of the code (as expressed in the soundness theorem in
Section~\ref{sec:soundness}), we again emphasize the \emph{compositional
  nature} of the way resource bounds are expressed in the types
themselves and in the typing judgments for process expressions.  Of
course, they reveal some intensional property of the implementations,
namely a bound on its cost, so different implementations of the same
plain session type may have different resource annotations.

The typing derivation provides a proof certificate on the resource
bound for a process.  For closed processes, usually typed as
\[
\cdot \entailpot{p} Q :: (c : \one)
\]
the number $p$ provides a worst case bound for the number of messages
sent during the computation of $Q$, which always ends with the
process sending $\mi{end}$ along $c$.


\vspace{-5pt}

\section{\langname{}} \label{sec:language}

We briefly introduce the linear, process-only fragment of
SILL~\cite{Toninho13esop,Pfenning15FOSSACS}, which integrates
functional and concurrent computation.  A program in SILL is a
collection of processes exchanging messages through channels.  A new
process is \emph{spawned} by invoking a process definition, which
also creates a fresh channel that is \emph{provided} by the new
process.  The process that invokes a process definition becomes the
\emph{client} of the new process, communicating with it according to
the session types of the channel.  The exacting nature of linear
typing provides strong guarantees, including session fidelity (a form
of preservation) and absence of deadlocks (a form of progress).

We present an overview of the session types in SILL with a brief
description of their communication protocol.  They follow the type
grammar below.
\[
\begin{array}{rcl}
S, T & ::= & V \mid \ichoice{l_i : S} \mid \echoice{l_i : S}
\mid S \lolli T \mid S \tensor T \mid \one
\end{array}
\]
$V$ denotes a type variable here.  Types may be defined mutually
recursively in a global signature, where type definitions are
constrained to be \emph{contractive}~\cite{Gay05acta}.  This allows us
to treat them equi-recursively, which means we can silently replace a
type variable by its definition for the purpose of type-checking.

Internal choice $S \oplus T$ and external choice $S \with T$ have been
generalized to n-ary labeled sums $\ichoice{l_i : S_i}_{i \in I}$ and
$\echoice{l_i : S_i}_{i \in I}$ (for some index set $I$)
respectively. As a provider of internal choice
$\ichoice{l_i : S_i}_{i \in I}$, a process can send one of the labels
$l_i$ to its client. As a dual, a provider of external choice
$\echoice{l_i : S_i}_{i \in I}$ receives one of the labels $l_i$ which
is sent by its client.
We require
external and internal choice to comprise at least one label,
otherwise there would exist a linear channel without
observable communication along it, which is computationally
uninteresting and would complicate our proofs.
The connectives $\tensor$ and $\lolli$ are used to send
channels via other channels. A provider of $S \tensor T$
sends a channel of type $S$ to its client and then
behaves as a provider of $T$. A provider
of $S \lolli T$ receives a channel of type $S$ from its
client. The type of the provider and client changes
consistently, and the process terms of a provider
and client come in matching pairs.
\begin{table}[t!]
\setlength{\tabcolsep}{4pt}
\begin{tabular}{| l l l l l |}
\hline
\textbf{Type} & \textbf{Contin-} & \textbf{Process Term} & \textbf{Contin-} & \multicolumn{1}{c|}{\textbf{Description}} \\
\textbf{(current)} & \textbf{uation} & \textbf{(current)} & \textbf{uation} & \\
\hline
$c : \ichoice{\pot{l_i}{\ignore{q_i}} : S_i}$ & $c : S_k$ & $\esendl{c}{l_k} \semi P$
& $P$ & \hspace{-1em}provider sends label $l_k$ along \\
& & & & \hspace{-1em}$c$ \ignore{with potential $q_k$} \\
& & $\ecase{c}{l_i}{Q_i}_{i \in I}$ & $Q_k$ & \hspace{-1em}client receives label $l_k$ along \\
& & & & \hspace{-1em}$c$ \ignore{with potential $q_k$} \\
\hline
$c : \echoice{\pot{l_i}{\ignore{q_i}} : S_i}$ & $c : S_k$ & $\ecase{c}{l_i}{P_i}_{i \in I}$
& $P_k$ & \hspace{-1em}provider receives label $l_k$ \\
& & & & \hspace{-1em}along $c$ \ignore{with potential $q_k$} \\
& & $\esendl{c}{l_k} \semi Q$ & $Q$ & \hspace{-1em}client sends label $l_k$ along $c$ \\
& & & & \hspace{-1em}\ignore{with potential $q_k$} \\
\hline
$c : S \tensorpot{\ignore{q}} T$ & $c : T$ & $\esendch{c}{w} \semi P$
& $P$ & \hspace{-1em}provider sends channel $w : S$ \\
& & & & \hspace{-1em}along $c$ \ignore{with potential $q$} \\
& & $\erecvch{c}{y} \semi Q_y$ & $\hspace{-1em}[w/y]Q_y$ & \hspace{-1em}client receives channel $w : S$ \\
& & & & \hspace{-1em}along $c$ \ignore{with potential $q$} \\
\hline
$c : S \lollipot{\ignore{q}} T$ & $c : T$ & $\erecvch{c}{y} \semi P_y$
& $\hspace{-1em}[w/y]P_y$ & \hspace{-1em}provider receives channel $w :$\\
& & & & \hspace{-1em}$S$ along $c$ \ignore{with potential $q$} \\
& & $\esendch{c}{w} \semi Q$ & $Q$ & \hspace{-1em}client sends channel $w : S$ \\
& & & & \hspace{-1em}along $c$ \ignore{with potential $q$} \\
\hline
$c : \pot{\one}{\ignore{q}}$ & $-$ & $\eclose{c}$
& $-$ & \hspace{-1em}provider sends $\mi{end}$ along $c$ \\
& & & & \hspace{-1em}\ignore{with potential $q$} \\
& & $\ewait{c} \semi Q$ & $Q$ & \hspace{-1em}client receives $\mi{end}$ along $c$ \\
& & & & \hspace{-1em}\ignore{with potential $q$} \\
\hline
\end{tabular}
\caption{Linear resource-aware session types}
\vspace{-2.8em}
\label{tab:language}
\end{table}

Formally, the syntax of process expressions of \langname{} is same as
in SILL.
\[
\begin{array}{rcl}
P, Q & ::= & \ecut{x}{\mathcal{X}}{y}{Q}
\mid \fwd{x}{y}
\mid \esendl{x}{l_k} \semi P
\mid \ecase{x}{l_i}{P} \\
& & \mid \esendch{x}{w}
\mid \erecvch{x}{y} \semi P
\mid \eclose{x}
\mid \ewait{x} \semi P
\end{array}
\]
The first term $\ecut{x}{\mathcal{X}}{y}{Q}$ invokes a process
definition $\mathcal{X}$ to spawn a new process $P$, which uses the
channels in $\overline{y}$ as a client and provides service along a
fresh channel substituted for $x$ in $Q$.
A forwarding process $\fwd{x}{y}$ (which provides channel $x$)
identifies channels $x$ and $y$ and terminates. The effect is that
clients of $x$ will afterwards communicate along $y$.  We saw an
example of its use in Figure~\ref{fig:binary-counter-val}.
The rest of the program constructs concern communication between two
processes and are guided by their corresponding session type.  Table
\ref{tab:language} provides an overview of session types, associated
process terms, and their operational description (ignore the portions
in red).  For each connective in Table~\ref{tab:language}, the first
line provides the perspective of the provider, while the second line
provides that of the client. The first two columns present the type of
the channel before (\textbf{current}) and after
(\textbf{continuation}) the interaction. Similarly, the next two
columns present the process terms before and after the
interaction. Finally, the last column presents the operational
description.

We conclude by illustrating the syntax, types and semantics
of SILL using a simple example. Recall the counter protocol
(ignoring the resource annotations in red):
\begin{center}
\begin{minipage}{3cm}
\begin{tabbing}
$\bits = \ichoice{\pot{\m{zero}}{\ignore{0}} : \bits, \pot{\m{one}}{\ignore{0}} : \bits, \pot{\m{\$}}{\ignore{0}} : \pot{\one}{\ignore{0}}}$ \\
$\ctr\ignore{[n]} = \echoice{\pot{\m{inc}}{\ignore{1}} : \ctr\ignore{[n+1]}, \pot{\m{val}}{\ignore{2\ceil{\log(n+1)}+2}} : \bits}$
\end{tabbing}
\end{minipage}
\end{center}
The type prescribes that a process providing service along
a channel of type $\ctr$ will either receive an $\m{inc}$
or a $\m{val}$ label. If it receives an $\m{inc}$ label,
the channel will recurse back to the $\ctr$ type.
If it receives a $\m{val}$ label, it will continue
by providing $\m{bits}$, sending a sequence of labels $\m{zero}$
and $\m{one}$ closed out with $\m{\$}$ and $\mi{end}$.

We present implementations of the $b0$, $b1$ and $e$ processes
respectively that were analyzed in Section~\ref{sec:overview} in
Figure~\ref{fig:bincount_impl}.  In the comments we show the types of
the channels after the interaction on each line (again ignoring the
annotations in red).
\begin{figure}[t]
\begin{ntabbing}
\reset
$(t : \ctr\ignore{[n]}) \entailpot{\ignore{0}} b0 :: (s : \ctr\ignore{[2n]})$ \label{valproc:b0_type}\\
 \quad \= $\procdef{b0}{t}{s} =$ \label{valproc:b0_def} \\
 \> \qquad $\m{case}\; s$
\= $(\m{inc} \Rightarrow$ \= $\procdef{b1}{t}{s}$
\hspace{1em} \= $\%\quad (t : \ctr\ignore{[n]}) \entailpot{\ignore{1}} s : \ctr\ignore{[2n+1]}$ \label{valproc:b0_1} \\
\>\> $ \mid \m{val} \Rightarrow$ \> $\esendl{s}{\m{zero}} \semi$
\> $\% \quad (t : \ctr\ignore{[n]}) \entailpot{\ignore{2\ceil{\log(2n+1)}+2-1}} s : \bits$ \label{valproc:b0_2} \\
\>\>\> $\esendl{t}{\m{val}} \semi$
\> $\% \quad (t : \bits) \entailpot{\ignore{2\ceil{\log(2n+1)} + 1 - 2\ceil{\log(n+1)} - 3}} s : \bits$ \label{valproc:b0_3} \\
\>\>\> $\fwd{s}{t})$
\> $\% \quad (t : \bits) \entailpot{\ignore{0}} s : \bits$ \label{valproc:b0_4}
\end{ntabbing}

\begin{ntabbing}
$(t : \ctr\ignore{[n]}) \entailpot{\ignore{1}} b1 :: (s : \ctr\ignore{[2n+1]})$ \label{valproc:b1_type}\\
 \quad \= $\procdef{b1}{t}{s} =$ \label{valproc:b1_def} \\
 \> \qquad $\m{case}\; s$
\= $(\m{inc} \Rightarrow$ \= $\esendl{t}{\m{inc}} \semi$
\hspace{3.5em} \= $\%\quad (t : \ctr\ignore{[n+1]}) \entailpot{\ignore{0}} s : \ctr\ignore{[2n+2]}$ \label{valproc:b1_1} \\
\>\>\> $\procdef{b0}{t}{s}$
\> $\%$ \label{valproc:b1_2}\\
\>\> $ \mid \m{val} \Rightarrow$ \> $\esendl{s}{\m{one}} \semi$
\> $\% \quad (t : \ctr\ignore{[n]}) \entailpot{\ignore{2\ceil{\log(2n+2)}+2-1}} s : \bits$ \label{valproc:b1_3} \\
\>\>\> $\esendl{t}{\m{val}} \semi$
\> $\% \quad (t : \bits) \entailpot{\ignore{2\ceil{\log(2n+2)} + 1 - 2\ceil{\log(n+1)} - 3}} s : \bits$ \label{valproc:b1_4} \\
\>\>\> $\fwd{s}{t})$
\> $\% \quad (t : \bits) \entailpot{\ignore{0}} s : \bits$ \label{valproc:b1_5}
\end{ntabbing}

\begin{ntabbing}
$\cdot \entailpot{\ignore{0}} e :: (s : \ctr\ignore{[0]})$ \label{valproc:end_type}\\
 \quad \= $\procdefna{e}{s} =$ \label{valproc:end_def}\\
 \> \qquad $\m{case}\; s$
\= $(\m{inc} \Rightarrow$ \= $\procdefna{e}{t} \semi$
\hspace{6em} \= $\%\quad (t : \ctr\ignore{[0]}) \entailpot{\ignore{1}} (s : \ctr\ignore{[1]})$ \label{valproc:end_1}\\
\> \> \> $\procdef{b1}{t}{s}$ \label{valproc:end_2}\\
\>\> $\mid \m{val} \Rightarrow$
\> $\esendl{s}{e} \semi$
\> $\% \quad \cdot \entailpot{\ignore{2\ceil{\log(0+1)} + 2 - 1}} s : \pot{\one}{0}$ \label{valproc:end_3} \\
\>\>\> $\eclose{s})$ \label{valproc:end_4}
\end{ntabbing}
\caption{Implementations for the $b0$, $b1$ and $e$ processes
with their type derivations demonstrating the binary counter.}
\vspace{-1.5em}
\label{fig:bincount_impl}
\end{figure}
Since the $b0$ process provides and external choice along
$s$, $b0$ needs to branch based on the label received
(line~\ref{valproc:b0_1}). If it receives the label $\m{inc}$,
the type of the channel updates to $\ctr$, as indicated
on the typing in the comment. At this point, we
can spawn the $b1$ process since the type on line
\ref{valproc:b0_1} matches with the type of the $b1$ process
(line~\ref{valproc:b1_type}). If instead $b0$ receives
the $\m{val}$ label along $s$, it continues at type $\m{bits}$.
It sends $\m{zero}$ (since the lowest bit is indeed zero).
It then requests the value of the higher bits by sending
$\m{val}$ along channel $t$.  Now both $s$ and $t$ have
type $\m{bits}$ (indicated in the typing on line
\ref{valproc:b0_2}) and $b0$ can terminate by forwarding
further communication along $s$ to $t$.

The $b1$ process operates similarly, taking care to handle the carry
upon increment by sending an $\m{inc}$ label along $t$.  The $e$
process spawns a new $e$ process and continues as $b1$ upon receiving
the label $\m{inc}$ and closes the channel after sending $\$$ when
receiving $\m{val}$.


\vspace{-5pt}

\section{Cost Semantics}\label{sec:cost}

We present an operational cost semantics for \langname{}
that tracks the total work performed by a system.
Like previous work, our semantics is a substructural operational
semantics~\cite{Pfenning09LICS} based on multiset
rewriting~\cite{Cervesato06RSP} and asynchronous
communication~\cite{Pfenning15FOSSACS}.
It can be seen as a combination of an asynchronous version of a
recently introduced synchronous session-type
semantics~\cite{Balzer17ICFP} with the cost tracking semantics of
Concurrent C0~\cite{SilvaFP17}.  The technical advantage of our
semantics is that it avoids the complex operational artifacts of Silva
et al.~\cite{SilvaFP17} such as message buffers: processes and
messages can be typed with exactly the same typing rules, changing
only the cost metric.

We will only count communication costs, ignoring internal.  To this
end, we introduce $3$ costs, $\mlab$, $\mchan$ and $\mcl$, for labels,
channels, and close messages, respectively. A concrete semantics can
be obtained by setting appropriate values for each of those
metrics. For instance, setting $\mlab = \mchan = \mcl = 1$ will lead
to counting the total number of messages exchanged.

Our cost semantics is asynchronous, that is, processes can continue
their evaluation without wait after sending a message. In order to
guarantee session fidelity the semantics must ensure that messages are
received in the order they are sent.  Intuitively, we can think of
channels as FIFO message buffers, although we will formally define
them differently.  Synchronous communication can be implemented in our
language in a type-safe, logically motivated manner exploiting adjoint
shift operators (see~\cite{Pfenning15FOSSACS}).


A collection of communicating process is called a
\emph{configuration}.  A configuration is formally modelled as a
multiset of propositions $\proc{c}{w, P}$ and $\msg{c}{w, M}$.  The
predicate $\proc{c}{w, P}$ describes a process executing process
expression $P$ and providing channel $c$.  The predicate
$\msg{c}{w, M}$ describes the message $M$ on channel $c$.  In order to
guarantee that messages are received in they order they are sent, only
\emph{a single message} can be on a given channel $c$.  In order for
computation to remain truly asynchronous, every send operation (except
for $\m{close}$) on a channel $c$ creates not only a fresh message,
but also a fresh continuation channel $c'$ for the next message.  This
continuation channel is encoded within the message via a forwarding
operation.  Remarkably, this simple device allows us to assign session
types to messages just as if they were processes!  Since $M$ need only
encode a message, it has a restricted grammar.
\[
\begin{array}{lcl}
M & ::= & \fwd{c}{c'} \mid \esendl{c}{l_k} \semi \fwd{c}{c'}
\mid \esendl{c}{l_k} \semi \fwd{c'}{c} \\
& & \mid \esendch{c}{e} \semi \fwd{c}{c'} \mid
\esendch{c}{e} \semi \fwd{c'}{c} \mid \eclose{c}
\end{array}
\]

The work is tracked by the local counter $w$
in the $\proc{c}{w, P}$
and $\msg{c}{w, M}$
propositions.  For a process $P$,
$w$
maintains the total work performed by $P$
so far.  When a process sends a message (i.e. creates a new $\m{msg}$
predicate), we increment its counter $w$
by the cost for sending. When a processes terminates we remove the
respective predicate from the configuration but need to preserve
the work done by the process. A process can terminate
either by sending a $\m{close}$
message, or by forwarding.  In either case, we can
conveniently preserve
the process' work in the $\m{msg}$
predicate to pass it on to the client process.

\begin{figure}[t!]
\[
\infer[\m{spawn}_c]
{\proc{c}{0, [c/x,\overline{e}/y]P_{x,\overline{y}}} \qquad \proc{d}{w, [c/x]Q_x}}
{\Sg(\mathcal{X}) = \espawn{x}{\mathcal{X}}{\overline{y}}{P_{x,\overline{y}}}
\qquad \proc{d}{w, \ecut{x}{\mathcal{X}}{e}{Q_x}} \qquad \fresh{c}}
\]
\[
\infer[\m{fwd}_s]
{\msg{c}{w, \fwd{c}{d}}}
{\proc{c}{w, \fwd{c}{d}}}
\]
\[
\infer[\m{fwd}^+_r]
{\proc{c}{w + w', [c/d]P}}
{\proc{d}{w, P} \qquad \msg{c}{w', \fwd{c}{d}}}
\]
\[
\infer[\m{fwd}^-_r]
{\proc{e}{w + w', [d/c]P_c}}
{\proc{e}{w, P_c} \qquad \msg{c}{w', \fwd{c}{d}}}
\]
\[
\infer[\oplus C_s]
{\proc{c'}{w+\mlab, [c'/c]P} \qquad \msg{c}{0, \esendl{c}{l_k} \semi \fwd{c}{c'}}}
{\proc{c}{w, \esendl{c}{l_k} \semi P} \qquad \fresh{c'}}
\]
\[
\infer[\oplus C_r]
{\proc{d}{w+w', [c'/c]Q_k}}
{\msg{c}{w, \esendl{c}{l_k} \semi \fwd{c}{c'}} \qquad \proc{d}{w', \ecase{c}{l_i}{Q_i}_{i \in I}}}
\]
\[
\infer[\tensor C_s]
{\proc{c'}{w+\mchan, [c'/c]P} \qquad \msg{c}{0, \esendch{c}{e} \semi \fwd{c}{c'}}}
{\proc{c}{w, \esendch{c}{e} \semi P} \qquad \fresh{c'}}
\]
\[
\infer[\tensor C_r]
{\proc{d}{w+w', [c'/c]Q_e}}
{\msg{c}{w, \esendch{c}{e} \semi \fwd{c}{c'}} \qquad \proc{d}{w', \erecvch{c}{x} \semi Q_x}}
\]
\[
\infer[\lolli C_s]
{\proc{d}{w+\mchan, [c'/c]P} \qquad \msg{c'}{0, \esendch{c}{e} \semi \fwd{c'}{c}}}
{\proc{d}{w, \esendch{c}{e} \semi P} \qquad \fresh{c'}}
\]
\[
\infer[\lolli C_r]
{\proc{c}{w+w', [c'/c]Q_e}}
{\msg{c'}{w, \esendch{c}{e} \semi \fwd{c'}{c}} \qquad \proc{c}{w', \erecvch{c}{x} \semi Q_x}}
\]
\[
\infer[\one C_s]
{\msg{c}{w+\mcl, \eclose{c}}}
{\proc{c}{w, \eclose{c}}}
\]
\[
\infer[\one C_r]
{\proc{d}{w + w', Q}}
{\msg{c}{w, \eclose{c}} \qquad \proc{d}{w', \ewait{c} \semi Q}}
\]
\caption{Cost semantics tracking total work for programs in SILL}
\vspace{-1.5em}
\label{fig:cost_semantics}
\end{figure}

The semantics is defined by a set of rules rewriting the configuration that
\emph{consume} the proposition in the premise of the rule and \emph{produce} the
propositions in the conclusion (rules should be read top-down!). A step
consists of non-deterministic application of a rule whose premises
matches a part of the configuration. Consider for instance the rule
$C_s$
that describes a client that sends label $l_k$ along channel $c$.
\[
\infer[\with C_s]
{\proc{d}{w+\mlab, [c'/c]P} \qquad \msg{c'}{0, \esendl{c}{l_k} \semi \fwd{c'}{c}}}
{\proc{d}{w, \esendl{c}{l_k} \semi P} \qquad \fresh{c'}}
\]
The rule can be applied to every proposition of the form
$\proc{d}{w, \esendl{c}{l_k} \semi P}$.
When applying the rule, we generate a fresh channel
continuation channel $c'$ and replace the premise by
propositions $\proc{d}{w+\mlab, [c'/c]P}$
and $\msg{c'}{0, \esendl{c}{l_k} \semi \fwd{c'}{c}}$.
The message predicate contains the process $\esendl{c}{l_k} \semi \fwd{c'}{c}$
which will eventually deliver the message to the provider along $c$
and will continue communication along $c'$ (which is achieved by $\fwd{c'}{c}$).
The work of the process is incremented by $\mlab$
to account for the sent message, while the work of the message is $0$.

Conversely, the rule $\with C_r$
defines how a provider receives a label $l_k$
along $c$.
\[
\infer[\with C_r]
{\proc{c}{w+w', [c'/c]Q_k}}
{\msg{c'}{w, \esendl{c}{l_k} \semi \fwd{c'}{c}} \qquad \proc{c}{w', \ecase{c}{l_i}{Q_i}_{i \in I}}}
\]
The rule replaces the $\m{msg}$
and $\m{proc}$ propositions in the configuration that
match the premises, with the single $\m{proc}$
proposition in the conclusion.  Since the provider receives the label $l_k$,
it continues as $Q_k$.
However, we replace $c$
with $c'$
in $Q_k$
since the forwarding $\fwd{c'}{c}$
in the message process tells us that the next message will arrive on
channel $c'$.
If there is any work $w$ encoded in the message, it is
transferred to the recipient.  This is somewhat more general
than necessary for this particular rule, since in the
current system the work $w$ in a label-sending message $c.l_k$
will always be 0.

The rest of the rules of cost semantics are given in
Figure~\ref{fig:cost_semantics}.  The rule $\m{spawn}_c$
describes the creation of a new channel $c$
along with a spawning new process $\mathcal{X}$
implemented by $P_c$. This implementation is looked
up in a signature for the semantics $\Sg$ (maps process
names to the implementation code).
The new process is spawned with $0$
work (as it has not sent any messages so far), while $Q_c$
continues with the same amount of work.
In the rule $\m{fwd}_s$ a forwarding process creates a
\emph{forwarding message} and terminates. The work carried by this
special message is the same as the work done by the process, now
defunct.  A forwarding message form does not carry any real
information (except for the work $w$!); it just serves to identify the
two channels $c$ and $d$.  In an implementation this could be as
simple as concatenating two message buffers. We therefore do not count
forwarding messages when computing the work. Another reason forward
messages are special is that unlike all other forms of messages, they
are neither prescribed by nor manifest in a channel's type.
In our formal rules, the forwarding message can be absorbed
either into the client ($\m{fwd}^+_r$) or provider ($\m{fwd}^-_r$),
in both cases preserving the total amount of work.

%
%

The rules of the cost semantics are successively applied to a
configuration until the configuration becomes empty or the
configuration is stuck and none of the rules can be applied. At any
point in this local stepping, the total work performed by the system
can be obtained by summing the local counters $w$
for each predicate in the configuration. We will prove in
Section~\ref{sec:soundness} that this total work can be upper bounded
by the initial potential of the configuration that is typed in our
resource-aware type system.


\vspace{-5pt}

\section{Type System}\label{sec:typing}
We now present the resource-aware type system of
our language which extends the linear-only
fragment of SILL~\cite{Toninho13esop,Pfenning15FOSSACS}
with resource annotations.  It is in turn
based on intuitionistic linear logic~\cite{GIRARD1987LinLog}
with sequents of the form
\[
A_1, A_2, \ldots, A_n \vdash A
\]
where $A_1, \ldots A_n$ are the linear antecedents and
$A$ is the succedent. Under the Curry-Howard isomorphism
for intuitionistic linear logic, propositions are related to
session types, proofs to processes, and cut reduction in
proofs to communication. Appealing to this
correspondence, we assign a process term $P$
to the above judgment and label each hypothesis as
well as the conclusion with a channel.
\[
(x_1 : A_1), (x_2 : A_2), \ldots, (x_n : A_n) \vdash P :: (x : A)
\]
The resulting judgment states that process $P$ provides
a service of session type $A$ along channel $x$, using
the services of session types $A_1, \ldots, A_n$
provided along channels $x_1, \ldots, x_n$ respectively.
The assignment of a channel to the conclusion is convenient
because, unlike functions, processes do not evaluate to
a value but continue to communicate along their providing
channel once they have been created. For the judgment to
be well-formed, all the channel names have to be distinct.
Whether a session type is used or provided is determined
by its positioning to the left or right, respectively, of the
turnstile.

\begin{figure}[t!]
\[
\infer[\oplus R_k]
{\Sg \semi \W \entailpot{q} (\esendl{x}{l_k} \semi P) :: (x : \ichoice{\pot{l_i}{r_i} : S_i}_{i \in I})}
{q \geq p + r_k + \mlab \quad \Sg \semi \W \entailpot{p} P :: (x : S_k) \quad (k \in I)}
\]

\[
\infer[\oplus L]
{\Sg \semi \W \; (x : \ichoice{\pot{l_i}{r_i} : S_i}_{i \in I}) \entailpot{q} \ecase{x}{l_i}{Q_i}_{i \in I} :: (z : U)}
{q + r_i \geq q_i \quad \Sg \semi \W \; (x : S_i) \entailpot{q_i} Q_i :: (z : U) \quad (\forall i \in I)}
\]


\[
\infer[\lolli R]
{\Sg \semi \W \entailpot{q} (\erecvch{x}{y} \semi P_y) :: (x : S \lollipot{r} T)}
{q + r \geq p \quad \Sg \semi \W \; (y : S) \entailpot{p} P_y :: (x : T)}
\]

\[
\infer[\lolli L]
{\Sg \semi \W \; (w : S) \; (x : S \lollipot{r} T) \entailpot{q} (\esendch{x}{w} \semi Q) :: (z : U)}
{q \geq p + r + \mchan \quad \Sg \semi \W \; (x : T) \entailpot{p} Q :: (z : U)}
\]

\[
\infer[\tensor R]
{\Sg \semi (w : S) \; \W \entailpot{q} \esendch{x}{w} \semi P :: (x : S \tensorpot{r} T)}
{q \geq p + r + \mchan \quad \Sg \semi \W \entailpot{p} P :: (x : T)}
\]

\[
\infer[\tensor L]
{\Sg \semi \W \; (x : S \tensorpot{r} T) \entailpot{q} \erecvch{x}{y} \semi Q_y :: (z : U)}
{q + r \geq p \quad \Sg \semi \W \; (y : S) \; (x : T) \entailpot{p} Q_y :: (z : U)}
\]

\[
\infer[\one R]
{\Sg \semi \cdot \entailpot{q} \eclose{x} :: (x : \pot{\one}{r})}
{q \geq r + \mcl}
\qquad
\infer[\one L]
{\Sg \semi \W \; (x : \pot{\one}{r}) \entailpot{q} \ewait{x} \semi Q :: (z : U)}
{q + r \geq p \quad \Sg \semi \W \entailpot{p} Q :: (z : U)}
\]

\[
\infer[\mu R]
{\Sg \semi \W \entailpot{q} P :: (x : V)}
{(V = S_V) \in \Sg \quad \Sg \semi \W \entailpot{q} P :: (x : S_V)}
\]

\[
\infer[\mu L]
{\Sg \semi \W \; (x : V) \entailpot{q} P :: (z : U)}
{(V = S_V) \in \Sg \quad \Sg \semi \W \; (x : S_V) \entailpot{q} P :: (z : U)}
\]
\caption{Typing rules for session-typed programs (remaining rules are given in the text)}
\vspace{-1.5em}
\label{fig:type_rules}
\end{figure}

Resource-aware session types are given by the following grammar.
\[
\begin{array}{rcl}
S, T & ::= & V \mid \ichoice{\pot{l_i}{q_i} : S} \mid \echoice{\pot{l_i}{q_i} : S}
\mid S \lollipot{q} T \mid S \tensorpot{q} T \mid \pot{\one}{q}
\end{array}
\]
Here, $V$ is a type variable.
The meaning of the types and the process terms associated
with it are defined in Table~\ref{tab:language} (annotations
and descriptions pertaining to potentials are marked in red).

The typing judgment of \langname{} has the form
$$\Sg ; \W \entailpot{q} P :: (x : S) \;.$$
Intuitively, the judgment describes a process in state
$P$ using the context $\W$ and signature $\Sg$
and providing service along channel
$x$ of type $S$. In other words, $P$ is the provider for
channel $x : S$, and a client for all the channels in $\W$.
The resource annotation $q$ is a natural number
and defines the potential stored in the process $P$.
$\Sg$ defines the signature containing type and process
definitions. These definitions are needed to typecheck processes
which refer to a type definition or spawn a new process.

The signature $\Sg$ is defined as a possibly infinite set of
type definitions $V = S_V$ and process definitions
$\espawn{x : S}{\mathcal{X} \fpot q}{\overline{y : W}}{P_{x, \overline{y}}}$.
The equation $V = S_V$ is used to define the type variable $V$ as $S_V$.
We treat such definitions \emph{equirecursively}.  For instance,
$\ctr[n] = \echoice{\pot{\m{inc}}{1} : \ctr[n+1], 
\pot{\m{val}}{2\ceil{\log (n+1)}+2} : \bits}$ exists in the
signature for all $n \in \mathbb{N}$ for the binary counter system.
Type families exist only at the meta-level and $\ctr[n]$ is
treated as a regular type variable. 
The process definition $\espawn{x : S}{\mathcal{X} \fpot q}{\overline{y : W}}{P_{x, \overline{y}}}$
defines a (possibly recursive) process named $\mathcal{X}$ that is implemented by $P_{x, \overline{y}}$
provides along channel $x : S$, and uses the channels $\overline{y : W}$ as a client.
The process also stores a potential $q$, shown as $\mathcal{X} \fpot q$
in the signature. 
For instance,
for the binary counter system,
$\espawn{s : \ctr[2n]}{b0 \fpot 0}{t : \ctr[n]}{P_{s, t}}$
($P_{s, t}$ defines the implementation of $b0$) exists
in the signature for all $n \in \mathbb{N}$.

Messages are typed differently from processes as
their work counters $w$ (introduced in the predicate
$\msg{c}{w, M}$) are not incremented when they
actually deliver the message to the receiver. Hence,
to type the messages,
we define an auxiliary cost-free typing judgment,
$\Sg ; \W \entailpotcf{q} P :: (x : S)$, which follows
the same typing rules as Figure~\ref{fig:type_rules},
but with $\mlab = \mchan = \mcl = 0$. This avoids
paying the cost for sending a message twice. A
fresh signature $\Sg$ is used in the derivation
of the cost-free judgment.

The idea of the type system is that each message carries potential and
the sending process pays the potential along with the cost of
sending a message from its local potential. The receiving process
receives the potential when it receives the message and adds it to its
local potential.  For example, consider the rule $\with L_k$ for
a client sending a label $l_k$ along channel $x$.
\[
\infer[\with L_k]
{\Sg \semi \W \; (x : \echoice{\pot{l_i}{r_i} : S_i}) \entailpot{q} \esendl{x}{l_k} \semi Q :: (z : U)}
{q \geq p + r_k + \mlab \quad \Sg \semi \W \; (x : S_k) \entailpot{p} Q :: (z : U)}
\]
Since the continuation $Q$
needs potential $p$
to typecheck, and the potential to be sent with the label is $r_k$,
we need a total potential of at least $p + r_k + \mlab$,
where $\mlab$
is the cost of sending a label. Hence, we get the constraint
$q \geq p + r_k + \mlab$.

The rule $\with R$
describes a provider that is awaiting a message on channel $x$
and has local potential $q$ available.
\[
\infer[\with R]
{\Sg \semi \W \entailpot{q} \ecase{x}{l_i}{P_i}_{i \in I} :: (x : \echoice{\pot{l_i}{r_i} : S_i})}
{q + r_i \geq q_i \quad \Sg \semi \W \entailpot{q_i} P_i :: (x : S_i) \quad (\forall i \in I)}
\]
The second premise tells us that the
branch $P_i$ needs potential $q_i$ to typecheck.
But the branch $P_i$ is reached after receiving the
label $l_i$ with potential $r_i$. Hence, the initial potential
$q$ must be able to cover the difference  $q_i - r_i$.
Since potential $q$ can typecheck all the branches, 
we get the constraint $q \geq q_i - r_i$ for all $i$.

To spawn a new process defined by $\mathcal{X}$, we split the context $\W$ into $\W_1 \; \W_2$,
and we use $\W_1$
to type the newly spawned process and $\W_2$ for the continuation $Q_x$.
\[
\sinfer{\m{spawn}}
{\Sg \semi \W_1 \; \W_2 \entailpot{r} (\ecut{x}{\mathcal{X}}{y}{Q_x}) :: (z : U)}
{r \geq p + q \quad \espawn{x' : S}{\mathcal{X} \fpot p}{\overline{y' : W}}{P_{x', \overline{y'}}} \in \Sg \\
\W_1 = \overline{y : W} \quad \Sg \semi \W_2 \; (x : S) \entailpot{q} Q_x :: (z : U)}
\]
If the spawned process needs potential $p$
(indicated by the signature) and the continuation needs potential
$q$
then the whole process needs potential $r \geq p + q$.

A forwarding process $\fwd{x}{y}$ terminates and its potential
$q$ is lost. Since we do not count forwarding messages in our
cost semantics, we don't need any potential to type the forward.
\[
\infer[\m{id}]
{\Sg \semi y : S \entailpot{q} \fwd{x}{y} :: (x : S)}
{q \geq 0}
\]
The rest of the rules are given in Figure~\ref{fig:type_rules}.
They are similar to the discussed rules
and we omit their explanation. 

As an illustration, the resource-aware type for the binary
counter was presented in Section~\ref{sec:language} (marked
in red). Also, Figure~\ref{fig:bincount_impl} provides the type
derivation of the $b0$, $b1$ and $e$ processes (again marked
in red).
The annotations, along with the type derivation,
prove that an
increment has an amortized resource cost of $1$ (potential
annotation of $\m{inc}$ in $\bits$ type) and reading a value has
a resource cost of $2 \ceil{\log (n+1)} + 2$.



\vspace{-5pt}

\section{Soundness} \label{sec:soundness}

This section concludes the discussion of \langname{} by proving the
soundness of the resource-aware type system with respect to the cost
semantics.
So far, we have analyzed and type-checked processes in isolation.
However, as our cost semantics indicates, processes always exist in a
configuration, where they interact with other processes. Hence, we
need to extend the typing rules to configurations.

\paragraph{Configuration Typing}
At runtime, a program state in \langname{} is a
set of processes interacting via messages.
Such a set is represented as as a multi-set of $\m{proc}$ and $\m{msg}$
predicates as described in Section~\ref{sec:cost}.
To type the resulting configuration $\config$,
we first need to define a well-formed signature.

A signature $\Sg$ is said to be \emph{well formed} if
every process definition
$\espawn{x : S}{\mathcal{X} \fpot p}{\overline{y : W}}{P_{x, \overline{y}}}$ in $\Sg$
is well typed according to the process typing judgment, i.e.
$\Sg \semi \W \entailpot{p} P_{x, \overline{y}} :: (x : S)$.

We use the following judgment to type a configuration.
\[
\Sg ; \W_1 \potconf{S} \config :: \W_2
\]
It states that $\Sg$ is well-formed
and that the configuration $\config$
uses the channels in the context $\W_1$ and provides
the channels in the context $\W_2$. The natural number $S$ denotes the
sum of the total potential and work done by the system.
We call $S$ the \measure{} of the configuration.

\begin{figure}[t]
\[
\infer[\m{emp}]
{\Sg \semi (\cdot) \potconf{0} (\cdot) :: (\cdot)}
{}
\qquad
\infer[\m{compose}]
{\Sg \semi \W \potconf{S + S'}(\config \; \config') :: \W''}
{\Sg \semi \W \potconf{S} \config :: \W' \qquad \Sg \semi \W' \potconf{S'} \config' :: \W''}
\]
\[
\infer[C_{\m{proc}}]
{\Sg \semi \W \; \W_1 \potconf{p + w}(\proc{x}{w, P}) :: (\W \; (x:A) )}
{\Sg \semi \W_1 \entailpot{p} P :: (x : A)}
\]
\[
\infer[C_{\m{msg}}]
{\Sg \semi \W \; \W_1 \potconf{p + w}(\msg{x}{w, P}) :: (\W \; (x:A) )}
{\Sg \semi \W_1 \entailpotcf{p} P :: (x : A)}
\]
\caption{Typing rules for a configuration}
\vspace{-1.5em}
\label{fig:config_typing}
\end{figure}

The configuration typing judgment is defined using
the rules presented in Figure~\ref{fig:config_typing}.
The rule $\m{emp}$  defines that an empty configuration
is well-typed with \measure{} $0$. The rule $\m{compose}$
composes two
configurations $\config$ and $\config'$: $\config$ provides
service on the channels in $\W'$ while $\config'$ uses
the channels in $\W'$. The \measure{} of the composed
configuration $\config \; \config'$ is obtained by summing
up their individual \measure{}s. The rule $C_{\m{proc}}$
creates a configuration out of a single process. The
\measure{} of this singleton configuration is obtained by
adding the potential of the process and the work performed
by it. Similarly, the rule $C_{\m{msg}}$  creates a configuration
out of a single message. Since we account for the cost while
typing the processes (see Figure~\ref{fig:type_rules}), using
the same judgment to type the processes would lead to
paying for the same message twice. Hence, the messages
are typed in a cost-free judgment where $\mlab = \mchan =
\mcl = 0$.

\paragraph{Soundness}
Theorem~\ref{thm:soundness} is the main theorem of the paper. It is a
stronger version of a classical type preservation theorem and the
usual type preservation is a direct consequence. Intuitively, it
states that the \measure{} of a configuration never increases during
an evaluation step. This also implies that the \measure{} of the
initial configuration is an upper bound on the \measure{} of any
configuration it can ever step to. The soundness connects the
potential with the work (i.e. the type system with the cost
semantics).

\begin{theorem}[Soundness]\label{thm:soundness}
Consider a well-typed configuration $\config$ w.r.t. a well-formed
signature $\Sg$ such that $\Sg ; \W_1 \potconf{S} \config :: \W_2$.
If $\config \step \config'$, then there exist $\W_1'$, $\W_2'$ and $S'$
such that $\Sg ; \W_1' \potconf{S'} \config' :: \W_2'$ and $S' \leq S$.
\end{theorem}

The proof of the soundness theorem is achieved
by a case analysis on the cost semantics, followed
by an inversion on the typing of a configuration.
The complete proof is presented in Appendix
\ref{app:soundness}. The preservation theorem
is a corollary of the soundness, since we prove
that the configuration $\config'$ is well-typed.

The soundness implies that the weight of a configuration is an upper
bound on the total work performed by an evaluation starting in that
configuration. We are particularly interested in the special case of a
configuration that starts with $0$
work. In this case, the weight corresponds to the initial potential of
the system.

\begin{corollary}[Upper Bound]
If $\Sg ; \W_1 \potconf{S}
\config :: \W_2$, and $\config \step^* \config'$,
then $S \geq W'$, where $W'$ is the total
work performed by the configuration $\config'$,
i.e. the sum of the work performed by each
process and message in $\config'$. In particular,
if the work done by the initial configuration $\config$ is $0$,
then the potential $P$ of the initial configuration
satisfies $P \geq W'$.
\end{corollary}

\begin{proof}
Applying the Soundness theorem successively,
we get that if $\config \step^* \config'$ and
$\Sg ; \W_1 \potconf{S} \config :: \W_2$ and
$\Sg ; \W_1' \potconf{S'} \config' :: \W_2'$,
then $S' \leq S$. Also, $S' = P' + W'$, where
$P'$ is the total potential of $\config'$, while
$W'$ is the total work performed so far in $\config'$.
Since $P' \geq 0$, we get that $W' \leq P' + W' =
S' \leq S$. In particular, if $W = 0$, we get that
$P = P + W = S \geq W'$, where $P$ and $W$
are the potential and work of the initial
configuration respectively.
\end{proof}

The progress theorem of \langname{} is a direct
consequence of progress in SILL~\cite{Toninho13esop}. Our cost semantics
are a cost observing semantics, i.e. it is just
annotated with  counters observing the work.
Hence, any runtime step that can be taken by a program
in SILL can be taken in \langname{}.


\vspace{-5pt}

\section{Case Study: Stacks and Queues}\label{sec:case_study}
As an illustration of our type system, we present a case
study on stacks and queues. Stacks and queues
have the same interaction protocol: they store elements of
a variable type $A$
and support inserting and deleting elements. They only
differ in their implementation and resource usage. We
express their common interface type as the simple
session type $\st{A}$.
\begin{center}
  \begin{minipage}{0.5\linewidth}
    \begin{tabbing}
      $\st{A} = \echoice{$ \= $\m{ins} : A \lolli \st{A},$ \\
        \> $\m{del} : \ichoice{\m{none} : \one, \m{some} : A \tensor \st{A}}}$ 
    \end{tabbing}
  \end{minipage}
\end{center}
The session type dictates that a process providing a service of
type $\st{A}$, gives a client the choice to either insert ($\m{ins}$)
or delete ($\m{del}$) an element of type A. Upon receipt of the label
$\m{ins}$, the providing process expects to receive a channel of
type $A$ to be enqueued and recurses. Upon receipt of the label
$\m{del}$, the providing process either indicates that the queue
is empty ($\m{none}$), in which case it terminates, or that there
is a channel stored in the queue ($\m{some}$), in which case it
deletes this channel, sends it to the client, and recurses.

To account for the resource cost, we need to add potential
annotations leading to two different resource-aware types
for stacks and queues. Since we are interested in counting
the total number of messages exchanged, we set
$\mlab = \mchan = \mcl = 1$ in our type system to obtain
a concrete bound. 

\paragraph{Stacks} The type for stacks is defined below.
\begin{center}
  \begin{minipage}{0.5\linewidth}
\begin{tabbing}
$\stack{A} = {\with} \{$ \= $\pot{\m{ins}}{0} : A \lollipot{0} \stack{A},$ \\
\> $\pot{\m{del}}{2} : {\oplus}\{ \pot{\m{none}}{0} : \pot{\one}{0}, 
\pot{\m{some}}{0} : A \tensorpot{0} \stack{A}\}\}$
\end{tabbing}
  \end{minipage}
\end{center}
A stack is implemented as a sequence of $\mi{elem}$ processes
terminated by an $\mi{empty}$ process. The implementation and
type derivation of $\mi{elem}$ is presented below.
\begin{ntabbing}
\reset
$(x{:}A)\; (t{:} \stack{A})\; \entailpot{0} \mi{elem} :: (s : \stack{A})$ \label{elem_type}\\
 \quad \= $\procdef{elem}{x \; t}{s} =$ \label{elem_def} \\
 \> \quad $\m{case}\; s$ \label{elem:1} \\
 \>\qquad
\= $(\m{ins} \Rightarrow$ \= $\erecvch{s}{y} \semi$
\hspace{4em} \= $\%\quad (y{:}A)\; (x{:}A)\; (t{:}\stack{A}) \entailpot{0} 
s: \stack{A}$ \label{elem:2} \\
 \>\>\> $\procdef{elem}{x \; t}{s'} \semi$ 
 \> $\%\quad (y{:}A)\; (s' : \stack{A}) \entailpot{0} s : \stack{A}$ \label{elem:3} \\
 \>\>\> $\procdef{elem}{y \; s'}{s}$ \label{elem:4} \\
 \>\> $\mid \m{del} \Rightarrow$ \> $\esendl{s}{\m{some}} \semi$ \> 
 $\%\quad (x{:}A)\; (t{:}\stack{A}) \entailpot{1} s : A \tensorpot{0} \stack{A}$ \label{elem:5} \\
 \>\>\> $\esendch{s}{x} \semi$
  \> $\%\quad t{:}\stack{A} \entailpot{0} s : \stack{A}$  \label{elem:6}\\
 \>\>\> $\fwd{s}{t})$ \label{elem:7}
\end{ntabbing}
The recursive $\mi{elem}$ process stores an element of the stack.
It uses channel $x : A$ (element being stored) and channel
$t : \stack{A}$ (tail of the stack) and provides service along
$s : \stack{A}$. The implementation demonstrates that if the
$\mi{elem}$ process receives an $\m{ins}$ message along $s$,
it receives the element $y$ (line~\ref{elem:2}),
spawns a new $\mi{elem}$ process using its original element
$x$ (line~\ref{elem:3}), and continues with another instance
of the $\mi{elem}$ process with the received element $y$ (line~\ref{elem:4}).
In this way, it adds the element $y$ to the head of the sequence.
Otherwise, $\mi{elem}$ receives a $\m{del}$ message along $s$ and
responds with the $\m{some}$ label (line~\ref{elem:5}), followed by the channel
$x$ it stores (line~\ref{elem:6}). It then forwards all communication along $s$
to $t$.

Inserting an element has no resource cost, since no messages are
sent by the $\mi{elem}$ process. Similarly, deleting an
element has a cost of $2$, which is used to send two messages: the
$\m{some}$ label and the element $x$. This is
reflected by the type $\stack{A}$, which needs $0$ and $2$
potential units for insertion and deletion, respectively, as indicated
by the resource annotations.

We now implement and type the $\mi{empty}$ process.
\begin{ntabbing}
$\cdot \entailpot{0} \mi{empty} :: (s: \stack{A})$ \label{empty:type}\\
 \quad \= $\procdefna{empty}{s} =$ \label{empty:def}\\
 \> \qquad $\m{case}\; s$
\= $(\m{ins} \Rightarrow$ \= $\erecvch{s}{y} \semi$
\hspace{3em} \= $\%\quad (y{:}A)\;  \entailpot{0} 
s: \stack{A}$ \label{empty:1}\\
 \>\>\> $\procdefna{empty}{e} \semi$ 
 \> $\%\quad (y{:}A)\; (e : \stack{A}) \entailpot{0} s : \stack{A}$ \label{empty:2}\\
 \>\>\> $\procdef{elem}{y \; e}{s}$ \label{empty:3}\\
 \>\> $\mid \m{del} \Rightarrow$ \> $\esendl{s}{\m{none}} \semi$ \> 
 $\%\quad \cdot \entailpot{1} s : \one$ \label{empty:4}\\
 \>\>\> $\eclose{s})$ \label{empty:5}
\end{ntabbing}
The sequence of $\mi{elem}$ processes ends with an $\mi{empty}$ process,
providing service along channel $s$ where it can receive the label $\m{ins}$
or $\m{del}$. If it receives the label $\m{ins}$, it
receives the element $y$ to be inserted (line~\ref{empty:1}),
spawns a new $\mi{empty}$ process (line~\ref{empty:2}),
and continues execution as an $\mi{elem}$ process with the
received element (line~\ref{empty:3}). On receiving the
label $\m{del}$, it just sends the $\m{none}$ label (line~\ref{empty:4})
followed by the $\m{close}$ message (line~\ref{empty:5}),
indicating that the stack is empty.

Inserting an element
sends no messages and thus has cost $0$. Deleting an element
sends two messages and has cost $2$, which is reflected in the
resource annotations of the labels in the type $\st{A}$. Note that
deleting an element
requires the system to send back two messages,
either the $\m{none}$ label followed by the close message,
or the $\m{some}$ label followed by the element.
Therefore, an implementation of stacks will have
a resource cost of at least $2$ for deletion.
This shows that the above implementation is the most efficient
w.r.t. our cost semantics because insertion
has no resource cost, and deletion has the least possible cost.

\paragraph{Queues} Next, we consider the queue interface which is
achieved by using the same $\st{A}$ interface
and annotating it with a different potential.
The tight potential bound depends on the
number of elements stored in the queue. Hence,
a precise resource-aware type needs access
to this internal measure in the type.
A type $\queue{A}[n]$ intuitively defines a queue of size $n$,
i.e. a process offering along a channel of type
$\queue{A}[n]$ connects a sequence of $n$ elements.
\begin{center}
  \begin{minipage}{0.5\linewidth}
\begin{tabbing}
$\queue{A}[n] = {\with} \{$ \= $\pot{\m{ins}}{2n} : A \lollipot{0} \queue{A}[n+1],$ \\
\> $\pot{\m{del}}{2} : {\oplus}\{ \pot{\m{none}}{0} : \pot{\one}{0}, 
\pot{\m{some}}{0} : A \tensorpot{0} \queue{A}[n-1]\}\}$
\end{tabbing}
  \end{minipage}
\end{center}
Similar to a stack, a queue is also implemented
by a sequence of $\mi{elem}$ processes,
connected via channels, and terminated by
the $\mi{empty}$ process.
We show the implementation of $elem$ below.
\begin{ntabbing}
\reset
$(x{:}A)\; (t{:} \queue{A}[n-1])\; \entailpot{0} \mi{elem} :: (s:\queue{A}[n])$ \label{q_elem:type} \\
 \; \= $\procdef{elem}{x \; t}{s} =$ \label{q_elem:def}\\
 \> \; $\m{case}\; s \; ($ \label{q_elem:1}\\
\> \quad \= $\m{ins} \Rightarrow$ \; \= $\erecvch{s}{y} \semi$
\hspace{0.2pt} \= $\%\; (y{:}A)\; (x{:}A)\; (t{:}\queue{A}[n-1]) \entailpot{2n} 
s {:} \queue{A}[n+1]$ \label{q_elem:2} \\
 \>\>\> $\esendl{t}{\m{ins}} \semi$ 
 \> $\%\; (y{:}A)(x {:} A) (t {:} A \lollipot{0} \queue{A}[n]) \entailpot{1} s {:} \queue{A}[n+1]$ \label{q_elem:3} \\
  \>\>\> $\esendch{t}{y} \semi$ 
 \> $\%\quad (x : A)\; (t : \queue{A}[n]) \entailpot{0} s : \queue{A}[n+1]$ \label{q_elem:4} \\
 \>\>\> $\procdef{elem}{x \; t}{s}$ \label{q_elem:5} \\
 \>\> $\mid \m{del} \Rightarrow$ \> $\esendl{s}{\m{some}} \semi$ \> 
 $\%\quad (x{:}A) (t{:}\queue{A}[n-1]) \entailpot{1} s {:} A \tensorpot{0} \queue{A}[n-1]$ \label{q_elem:6} \\
 \>\>\> $\esendch{s}{x} \semi$
  \> $\%\quad t{:}\queue{A}[n-1] \entailpot{0} s : \queue{A}[n-1]$ \label{q_elem:7} \\
 \>\>\> $\fwd{s}{t}$) \label{q_elem:8}
\end{ntabbing}
Similar to the implementation of a stack, the $\mi{elem}$
process provides along $s : \queue{A}$, stores the
element $x : A$, and uses the tail of the queue $t : \queue{A}$.
When the $\mi{elem}$ process receives the $\m{ins}$
message along $s$, it receives the element $y$ (line~\ref{q_elem:2}),
and passes the $\m{ins}$ message (line~\ref{q_elem:3})
along with $y$ (line~\ref{q_elem:3}) to $t$.
Since the process at the other end of $t$ is also implemented
using $\mi{elem}$, it passes along the element to its tail
too. Thus the element travels to the end
of the queue where it is finally inserted.
The deletion is similar to that for stack.

For each insertion, the $\m{ins}$ label along with the
element travels to the end of the queue.
Hence, the resource cost of each insertion
is $2n$ where $n$ is the size of the queue and
this is reflected in the type $\queue{A}$ in
the potential annotation of $\m{ins}$ as $2n$.
Similar to the stack, deletion has a resource cost of $2$
to get back the $\m{some}$ label and the element.
We now consider the $empty$ process. 
\begin{tabbing}
$\cdot \entailpot{0} \mi{empty} :: (s:\queue{A}[0])$ \\
 \quad \= $\procdefna{empty}{s} =$ \\
 \> \quad $\m{case}\; s \; ($ \\
\> \quad \= $\m{ins} \Rightarrow$ \; \= $\erecvch{s}{y} \semi$
\hspace{3em} \= $\%\quad (y{:}A)\;  \entailpot{0} 
s: \queue{A}[1]$ \\
 \>\>\> $\procdefna{empty}{e} \semi$ 
 \> $\%\quad (y{:}A)\; (e : \queue{A}[0]) \entailpot{0} s : \queue{A}[1]$ \\
 \>\>\> $\procdef{elem}{y \; e}{s}$ \\
 \>\> $\mid \m{del} \Rightarrow$ \> $\esendl{s}{\m{none}} \semi$ \> 
 $\%\quad \cdot \entailpot{1} s : \one$ \\
 \>\>\> $\eclose{s})$
\end{tabbing}
The implementation of $\mi{empty}$ process is identical to that of stacks.
Since insertion does not cause the process to send any
messages, its resource cost is $0$.  On the other hand, deletion
costs $2$ units because the process sends back the $\m{none}$
label followed by the $\m{close}$ message. This is correctly
reflected in the queue type. Since $s : \queue{A}[0]$,
the annotation for $\m{ins}$ is $2n = 2 \cdot 0 = 0$.
Similarly, $\m{del}$ is annotated with a potential of $2$.

The resource-aware types show that the implementation
for stacks is more efficient than that of queues. This 
follows from the potential annotation. The label $\m{ins}$
is annotated by $2n$ for $\stack{A}$ and with  $0$ for $\queue{A}$.
The label $\m{del}$ has the same annotation in both types.

\paragraph{Functional queues}
In a functional language, a queue is often implemented with two
lists. The idea is to enqueue into the first list and to dequeue from
the second list. If the second list is empty then we copy the first
list over, thereby reversing its order. Since the cost of the dequeue
operation varies drastically between the dequeue operations, amortized
analysis is again instrumental in the analysis of the worst-case
behavior and shows that the worst-case amortized cost for deletion
is actually a constant.

Appendix~\ref{subsec:app_queue} contains an implementation
of a functional queue in \langname{}. The type of the queue is
\begin{center}
  \begin{minipage}{0.5\linewidth}
\begin{tabbing}
$\queue{A}' = \echoice{$ \= $\pot{\m{ins}}{6} : A \lollipot{0} \queue{A}',$ \\
\> $\pot{\m{del}}{2} : \ichoice{$ \= $\pot{\m{some}}{0} : A \tensorpot{0} \queue{A}', \pot{\m{none}}{0} : \pot{\one}{0}}}$
\end{tabbing}
  \end{minipage}
\end{center}
Resource-aware session types enable us to translate the amortized
analysis to the distributed setting.
The type prescribes that an insertion has an amortized cost of $6$
while the deletion has an amortized cost of $2$. The main idea here
is that the elements are inserted with a constant potential in the
first list. While deleting, if the second list is empty, then this stored
potential in the first list is used to pay for copying the elements
over to the second list. The exact potential annotations for the
two lists can be found in Appendix~\ref{subsec:app_queue}. As
demonstrated from the resource-aware type, this implementation
is more efficient than the previous queue implementation, which
has a linear resource cost for insertion.

\paragraph{Generic clients}

The notion of efficiency of a store can be generalized and quantified
by considering clients for the stack and queue interface. A client
interacts with a generic store via a sequence of insertions and
deletions.  A provider can then implement the store as a stack, queue,
priority queue, etc. (same interface) and just expose the
resource-aware type for $\st{A}$.
Our type system can then use just the interface type and the generic client
implementation to derive resource bounds on the client. For
simplicity, the clients are typed in an affine type system which
allows us to throw away dummy channels (see below).

We provide a general mechanism for implementing clients
for a generic store. We define a generic $\st{A}$ type at which the
potential annotations are arbitrary natural numbers.
\begin{center}
  \begin{minipage}{0.5\linewidth}
\begin{tabbing}
$\st{A}[n] = {\with} \{$ \= $\pot{\m{ins}}{i} : A \lollipot{a} \st{A}[n+1],$ \\
\> $\pot{\m{del}}{d} : {\oplus}\{ \pot{\m{none}}{p} : \pot{\one}{e}, 
\pot{\m{some}}{s} : A \tensorpot{t} \st{A}[n-1]\}\}$
\end{tabbing}
  \end{minipage}
\end{center}
A client is defined by a list $\ell$ of $\m{ins}$ and $\m{del}$
messages that it sends to the store. We index the client $C_{\ell, n}$
using $\ell$, and the internal measure $n$ of the $\st{A}$ type.
The channel along which the client provides
is irrelevant for our analysis and is represented
using a dummy channel $d : D$. For ease of notation,
we define the potential needed for a client
$C_{\ell, n}$ as a function $\potf{\ell, n}$.

We implement the client $C_{\ell, n}$ as follows.
First, consider the case when $\ell = \nil$, i.e. an
empty list. 
\begin{eqnarray*}
\cdot \entailpot{0} C_{\nil, n} :: (d : D) \\
\procdefna{C_{\nil, n}}{d} = \eclose{d}
\end{eqnarray*}
The client for an empty list just closes the channel $d$.
We assume that all clients are typed with the cost-free metric to only
count for the messages sent inside the stores. So $C_{[], 0}$
needs $0$ potential.
For the potential function, this means
$\potf{\nil, n} = 0$.

Next, we implement the client when the head
of the list $\ell$ is $\m{ins}$.

\begin{tabbing}
$\W \; (x : A) \; (s : \st{A}[n]) \entailpot{q} C_{\m{ins}::\ell, n} :: (d : D)$ \\
\quad \= $\procdef{C_{\m{ins}::\ell, n}}{\W \; x \; s}{d} = $ \\
\> \quad \= $\esendl{s}{\m{ins}} \semi$
\hspace{6em} \= $\% \quad \W \; (x : A) \; (s : A \lollipot{a} \st{A}[n+1]) \entailpot{q-i} d : D$ \\
\>\> $\esendch{s}{x} \semi$
\> $\% \quad \W \; (s : \st{A}[n+1]) \entailpot{q - i - a} d : D$ \\
\>\> $\procdef{C_{\ell, n+1}}{\W \; s}{d}$
\end{tabbing}
The client sends an $\m{ins}$ label followed
by the element $x$. If $C_{\m{ins}::\ell, n}$ needs a potential
$q$, then the type derivation informs us that $C_{\ell, n+1}$
needs a potential $q - i - a$. Thus,
$\potf{\m{ins}::\ell, n} = \potf{\ell, n+1} + i + a$. Finally,
we show the client implementation if the head of the list
$\ell$ is $\m{del}$.

\begin{tabbing}
$\W \; (s : \st{A}[n]) \entailpot{q} C_{\m{del}::\ell} :: (d : D)$ \\
\quad \= $\procdef{C_{\m{del}::\ell, n}}{\W \; s}{d} = $ \\
\> \quad \= $\esendl{s}{\m{del}} \semi$
\hspace{2em} \= $\% \quad \W \; (s : \ichoice{\pot{\m{none}}{n} : \pot{\one}{e}, \pot{\m{some}}{s} : A \tensorpot{t} \st{A}[n-1]}) \entailpot{q-d} d : D$ \\
\>\> $\m{case} \; s \; ($ \\
\>\> $\m{some} \Rightarrow$ \; \=
$\erecvch{s}{x} \semi$
\hspace{1em} \= $\% \quad \W \; (x : A) \; (s : \st{A}[n-1]) \entailpot{q-d+s+t} d : D$ \\
\>\>\> $\procdef{C_{\ell, n-1}}{\W \; x \; s}{d}$ \\
\>\> $\mid \m{none} \Rightarrow$
\> $\ewait{s})$
\> $\% \quad \W \entailpot{q-d+p+e} d : D$
\end{tabbing}
The client sends the $\m{del}$ label and then
case analyzes on the label it receives. If it receives
the $\m{some}$ label, it receives the element and
then continues with $C_{\ell, n-1}$, else it
receives the $\m{none}$ label and waits for
the channel $s$ to close. In terms of the potential
function, this means
\[
\label{eqn:del_pot}
\potf{\m{del}::\ell, n} =
\begin{cases}
\potf{\ell, n-1} + d - s - t & \qquad \text{if} \quad n > 0 \\
\max(0, d - p - e) & \qquad \text{otherwise}
\end{cases}
\]
Walking through the list $\ell$ and chaining the
potential equations together, we can achieve a resource
bound on the client $C_{\ell, n}$ by computing
$\potf{\ell, n}$.

The $\stack{A}$ and $\queue{A}$ interface types
are specific instantiations of the $\st{A}$ type.
For the stack interface, plugging in appropriate
potential annotations
$i = a = p = e = s = t = 0$, and $d = 2$, we get
(ignoring the case where the stack becomes empty)
\[
\potf{\nil, n} = 0
\qquad
\potf{\m{ins}::\ell, n} = \potf{\ell, n+1}
\qquad
\potf{\m{del}::\ell, n} = \potf{\ell, n-1} + 2
\]
Similarly, considering the queue type as another instantiation
of the $\st{A}$ type, and plugging $a = p = e = s = t = 0$,
$i = 2n$ and $d = 2$, we get
\[
\potf{\nil, n} = 0
\qquad
\potf{\m{ins}::\ell, n} = \potf{\ell, n+1} + 2n
\qquad
\potf{\m{del}::\ell, n} = \potf{\ell, n-1} + 2
\]
Finally, looking at $\queue{A}'$ as another instantiation
of the $\st{A}$ type and plugging $a = p = e = s = t = 0$,
$i = 6$ and $d = 2$, we get
\[
\potf{\nil, n} = 0
\qquad
\potf{\m{ins}::\ell, n} = \potf{\ell, n+1} + 6
\qquad
\potf{\m{del}::\ell, n} = \potf{\ell, n-1} + 2
\]
This allows us to compare arbitrary clients of
two (same or different) interfaces and compare
their resource cost. The resource-aware types are
expressive enough to obtain these resource bounds
without referring the implementation of the store
interface. For instance,
an important property of queues is that every insertion is more
costly than the previous one. The cost of insertion depends
on the size of the queue, which, in turn, increases with every
insertion. Hence, the complexity of the
queue system depends on the sequence in which inserts
and deletes are performed. In particular, we can
consider the efficiency of two different clients for the
queue system, by solving the above system of equations.

For instance, consider two clients $Q_{\ell_1, n}$
and $Q_{\ell_2, n}$,
with two different message lists
$\ell_1 = [\m{ins}, \ldots, \m{ins}, \m{del},
\ldots, \m{del}]$, i.e. $m$
insertions followed by $m$
deletions, and
$\ell_2 = [\m{ins}, \m{del}, \m{ins}, \m{del}, \ldots, \m{ins},
\m{del}]$, i.e. $m$
instances of alternate insertions and deletions.  In both cases, we
have the same number of insertions and deletions. However, the
resource cost of the two systems are completely different.
Solving the system of equations,
we get that $\potf{\ell_1, n} = 2mn + m(m-1) + 2m$,
while $\potf{\ell_2, n} = 2m(n+1)$, which shows that
the second client is an order of magnitude more efficient
than the first one.
More examples are presented
in Appendix~\ref{sec:more}.

\vspace{-5pt}

\section{Related Work}
\label{sec:related}

Session types have been introduced by Honda~\cite{Honda93CONCUR,Honda98esop}.
The technical development in this article is based on previous work on
\cite{Toninho13esop,Pfenning15FOSSACS}.
By removing the potential annotation from the type
rules in Section~\ref{sec:typing} we arrive at the type system of
loc.\ cit. The internal measures and type families that we use are
inspired by~\cite{Griffith13nasa}. Other recent
innovations in session types include sharing of resources~\cite{Balzer17ICFP}
and dynamic monitors~\cite{JiaPOPL16}.
In contrast to our work, all the aforementioned articles do not
discuss static resource analysis.

Static resource bound analysis for sequential programs has been
extensively studied. Successful approaches are based on refinement
types~\cite{CicekGA15,CicekBGGH16}, linear dependent
types~\cite{LagoP13}, abstract
interpretation~\cite{Zuleger11,GulwaniMC09,AliasDFG10,CernyHKRZ15},
deriving and solving recurrence
relations~\cite{DannerLR15,FloresH14,AlbertGM12,KincaidBBR2017}, term
rewriting~\cite{BrockschmidtEFFG14,AvanziniLM15}.
These works do not consider message-passing programs nor concurrent or
parallel evaluation.

Our work is based on type-based amortized resource analysis.
Automatic amortized resource analysis (AARA) has been introduced as a type system to
automatically derive linear~\cite{Jost03} and polynomial
bounds~\cite{HoffmannW15} for sequential functional programs. It can
also be integrated with program logics to derive bounds for imperative
programs~\cite{Atkey10,CarbonneauxHRS17}. Moreover, it has been used
to derive bounds for term-rewrite systems~\cite{HofmannM15} and
object-oriented programs~\cite{Jost06}. A recent work also considers
bounds on the parallel evaluation cost (also called \emph{span}) of
functional programs~\cite{HoffmannS15}. The innovation of our work is
the integration of AARA and session types and the analysis of
message-passing programs that communicate with the outside
world. Instead of function arguments, our bounds depend on the
messages that are send along channels. As a result, the formulation
and proof of the soundness theorem is quite different from the
soundness of sequential AARA.

We are only aware of a couple of other works that study resource
bounds for concurrent programs. Gimenez et al.~\cite{GimenezM16}
introduced a technique for analyzing the parallel and sequential space
and time cost of evaluating interaction nets. While it also based on
linear logic and potential annotations, the flavor of the analysis is
quite different. Interaction nets are mainly used to model parallel
evaluation while session types focus on the interaction of
processes. A main innovation of our work is that processes can
exchange potential via messages. It is not clear how we can represent
the examples we consider in this article as interaction nets.
Albert et al.~\cite{AlbertFGM16,AlbertACGGMPR15} have studied
techniques for deriving bounds on the cost of concurrent programs that
are based on the actor model. While the goals of the work are similar
to ours, the used technique and considered examples are dissimilar. A
major difference is that our method is type-based and compositional.
A unique feature of our work is that
types describe bounds as functions of the messages that are sent along
a channel.


\vspace{-5pt}

\section{Conclusion}
\label{sec:concle}

We have introduced resource-aware session types, a linear type system
that combines session types~\cite{Honda93CONCUR,Pfenning15FOSSACS} and
type-based amortized resource analysis~\cite{Jost03,HoffmannW15} to
reason about the resource usage of message-passing processes.
The soundness of the type system has been proved for a core
session-typed language with respect to a cost semantics that tracks the total
communication cost in a system of processes. We have demonstrated that
our technique can be used to prove tight resource bounds and supports
amortized reasoning by analyzing standard session-type data structures
such as distributed binary counters, stacks, and queues.

Our approach addresses some of the main challenges of analyzing
message-passing programs such as compositionality and description of
symbolic bounds. However, there are several open problems that we
plan to tackle as part of future work. The technique we have developed in
this paper does not yet account for the concurrent execution cost of
processes, or the \emph{span}. We are working on a companion paper that
describes a type-based analysis to derive bounds on the span; the
earliest time a concurrent computation terminates assuming an infinite
number of processors. Due to data dependencies in a concurrent
program, a process needs to wait for messages from other
processes, and computing these waiting times statically 
makes span analysis challenging.

Similarly, we have focused on the foundations and meta theory of
resource-aware session types in this paper. The next step is to
implement our analysis. An advantage of our method is that it
is based upon type-based amortized resource analysis for
sequential programs.
We will integrate the type system with SILL
for functional programs~\cite{HoffmannW15}.
We designed the type system with automation in mind and we are
confident that we can support automatic type inference
using templates and LP solving
similar to AARA~\cite{Jost03,HoffmannW15}. To this end, we are working
on an algorithmic version of the declarative type system presented
here. 


\vspace{-5pt}

\bibliographystyle{splncs03}
\bibliography{refs,resource,hoffmann}

\appendix

\section{Proof of Soundness Theorem}\label{app:soundness}
We present the complete proof of the soundness
theorem. We reiterate the rules of cost semantics
and typing.
Figure~\ref{fig:cost_semantics} presents the cost
semantics for our session-typed language, while
Figure~\ref{fig:type_rules} presents the full set of
typing rules. Finally, Theorem~\ref{thm:soundness}
defines the soundness theorem establishing
that the typing rules for a configuration presented
in Figure~\ref{fig:config_typing} are sound w.r.t. the
rules for cost semantics presented in Figure
\ref{fig:cost_semantics}. We follow the complete
proof of the soundness theorem.

\begin{proof}
The proof proceeds by case analysis on the cost semantics
of our language, i.e. on the judgment $\config \step \config'$.
By the $\m{compose}$ rule, we can split the configuration
such that $\config = (\config_M \; \dc)$ and $\config' =
(\config_M' \; \dc)$ and $\config_M \step \config_M'$.
Using the $\m{compose}$ rule,
\[
\infer[\m{compose}]
{\Sg \semi \W \potconf{S_M + S_D}(\config_M \; \dc) :: \W''}
{\Sg \semi \W \potconf{S_M} \config_M :: \W' \quad \Sg \semi \W' \potconf{S_D} \dc :: \W''}
\]
\[
\infer[\m{compose}']
{\Sg \semi \W \potconf{S_M' + S_D}(\config_M' \; \dc) :: \W''}
{\Sg \semi \W \potconf{S_M} \config_M' :: \W' \qquad \Sg \semi \W' \potconf{S_D} \dc :: \W''}
\]
Hence, to show that $S' \leq S$, it suffices to show that
$S_M' \leq S_M$. We proceed by case analysis on the
$\config_M \step \config_M'$ judgment.
\begin{itemize}
\item Case ($\m{spawn}_c$) :
$\config_M = \proc{d}{w, \ecut{x}{P_x}{y}{Q_x}}$.
Inverting the typing rule $\m{spawn}$ on configuration $\config_M$,
we get
\begin{equation}\label{eqn:typ_spawn}
r \geq p + q \qquad 
\W_1 \; \W_2 \entailpot{r} \ecut{x}{P_x}{y}{Q_x} :: (d : U)
\end{equation}
and in $\config_M' = \proc{c}{0, P_c} \;
\proc{d}{w, Q_c}$, we get (the premise due to inversion)
\[
\W_1 \entailpot{p} P_c :: (c : S) \qquad 
\W_2 \; (c : S) \entailpot{q} Q_c :: (d : U)
\]
From the cost semantics rule $\m{spawn}_c$,
\[
\infer[\m{spawn}_c]
{\proc{c}{0, P_c} \qquad \proc{d}{w, Q_c}}
{\proc{d}{w, \ecut{x}{P_x}{y}{Q_x}}}
\]
Since $S$ is the sum of work and potential
of each process in the configuration, we get
$S_M = r + w$, while $S_M' = (p_P + w_P)
+ (p_Q + w_Q) = (p + 0) + (q + w) = p + q + w
\leq r + w = S_M$ since $p + q \leq r$ (by Equation
\ref{eqn:typ_spawn}).

\item Case ($\m{fwd}_s$) : $\config_M = \proc{c}{w, \fwd{c}{d}}$.
Inverting the typing rule $\m{fwd}$ on $\config_M$,
\[
q \geq 0 \qquad d : S \entailpot{q} \fwd{c}{d} :: (c : S)
\]
Inverting the same rule in $\config_M'$,
\[
q' \geq 0 \qquad d : S \entailpotcf{q'} \fwd{c}{d} :: (c : S)
\]
Using the semantics rule $\m{fwd}_s$,
\[
\infer[\m{fwd}_s]
{\msg{c}{w, \fwd{c}{d}}}
{\proc{c}{w, \fwd{c}{d}}}
\]
Since we are free to choose $q'$, we set $q' \leq q$.
Computing $S_M$ and $S_M'$, we get
\begin{eqnarray*}
S_M' & = & q' + w \\
& \leq & q + w \\
& = & S_M
\end{eqnarray*}

\item Case ($\m{fwd}^+_r$) : $\config_M = \proc{d}{w, P} \; \msg{c}{w', \fwd{c}{d}}$.
Inverting the $\m{fwd}$ rule for $\config_M$,
\begin{equation}\label{eqn:fwd_r_1}
q_1 \geq 0 \qquad \W \entailpot{q_1} P :: (d : S)
\end{equation}
\[
q_2 \geq 0 \qquad d : S \entailpot{q_2} \fwd{c}{d} :: (c : S)
\]
Using Equation~\ref{eqn:fwd_r_1} and noting that
$c$ and $d$ have the same type $S$, we get for $\config_M'$,
\[
q_1' \geq 0 \qquad \W \entailpot{q_1'} [c/d]P :: (c : S)
\]
From the cost semantics rule $\m{fwd}^+_r$, we get
\[
\infer[\m{fwd}^+_r]
{\proc{c}{w + w', [c/d]P}}
{\proc{d}{w, P} \qquad \msg{c}{w', \fwd{c}{d}}}
\]
Since we are free to choose $q_1'$, we set $q_1' \leq q_1$.
Computing $S_M$ and $S_M'$,  we get
\begin{eqnarray*}
S_M' & = & q_1' + w + w' \\
& \leq & q_1 + q_2 + w + w' \\
& = & (q_1 + w) + (q_2 + w') \\
& = & S_M
\end{eqnarray*}

\item Case ($\m{fwd}^-_r$) : $\config_M = \proc{e}{w, P} \; \msg{c}{w', \fwd{c}{d}}$.
Inverting the $\m{fwd}$ rule for $\config_M$,
\begin{equation}\label{eqn:fwd_r_2}
q_1 \geq 0 \qquad \W \; (c : S) \entailpot{q_1} P :: (e : U)
\end{equation}
\[
q_2 \geq 0 \qquad d : S \entailpot{q_2} \fwd{c}{d} :: (c : S)
\]
Using Equation~\ref{eqn:fwd_r_2} and noting that
$c$ and $d$ have the same type $S$, we get for $\config_M'$,
\[
q_1' \geq 0 \qquad \W \; (d : S) \entailpot{q_1'} [d/c]P :: (e : U)
\]
From the cost semantics rule $\m{fwd}^-_r$, we get
\[
\infer[\m{fwd}^-_r]
{\proc{e}{w + w', [d/c]P}}
{\proc{e}{w, P} \qquad \msg{c}{w', \fwd{c}{d}}}
\]
Since we are free to choose $q_1'$, we set $q_1' \leq q_1$.
\begin{eqnarray*}
S_M' & = & q_1' + w + w' \\
& \leq & q_1 + q_2 + w + w' \\
& = & (q_1 + w) + (q_2 + w') \\
& = & S_M
\end{eqnarray*}

\item Case ($\oplus C_s$) : $\config_M = \proc{c}{w, \esendl{c}{l_k} \semi P}$.
Inverting the typing rule $\oplus R_k$ on $\config_M$,
we get
\begin{equation}\label{eqn:typ_oplus_R_k_1}
q_1 \geq p + r_k + \mlab \qquad
\W \entailpot{q_1} (\esendl{c}{l_k} \semi P) :: (c : \ichoice{\pot{l_i}{r_i} : S_i}_{i \in I})
\end{equation}
and in $\config_M'$, we get (the premise due to inversion)
\[
\W \entailpot{p} [c'/c]P :: (c' : S_k)
\]
\begin{equation}\label{eqn:typ_oplus_fwd_oplus_R_k_2}
q_1' \geq q_2 + r_k \qquad c' : S_k \entailpot{q_1'} (\esendl{c}{l_k} \semi \fwd{c}{c'}) :: (c : \ichoice{\pot{l_i}{r_i} : S_i}_{i \in I})
\end{equation}
\begin{equation}\label{eqn:typ_oplus_fwd_oplus_R_k}
q_2 \geq 0 \qquad c' : S_k \entailpot{q_2} \fwd{c}{c'} :: (c : S_k)
\end{equation}
Using the cost semantics rule $\oplus C_s$, we get
\[
\infer[\oplus C_s]
{\proc{c'}{w+\mlab, [c'/c]P} \qquad \msg{c}{0, \esendl{c}{l_k} \semi \fwd{c}{c'}}}
{\proc{c}{w, \esendl{c}{l_k} \semi P}}
\]
Again, $S_M' = (p + w + \mlab) + (q_1' + 0)$.
Since, we need to prove that there exists such an $S_M'$,
we can choose $q_1'$ and $q_2$ arbitrarily such that
they satisfy Equations \ref{eqn:typ_oplus_fwd_oplus_R_k_2}
and \ref{eqn:typ_oplus_fwd_oplus_R_k}. We set
$q_2 = 0$ and $q_1' = r_k$. Hence,
$S_M' = (p + w + \mlab) + r_k \leq q_1 + w
\leq S_M$ (by Equation~\ref{eqn:typ_oplus_R_k_1}) 

\item Case ($\oplus C_r$) : $\config_M = \msg{c}{w, \esendl{c}{l_k} \semi \fwd{c}{c'}} \; \proc{d}{w', \ecase{c}{l_i}{Q_i}_{i \in I}}$.
Inverting the typing rule $\oplus L$ on $\config_M$,
we get
\begin{equation}\label{eqn:oplus_C_s_1}
q_1 \geq q_2 + r_k \qquad c' : S_k \entailpot{q_1} (\esendl{c}{l_k} \semi \fwd{c}{c'}) :: (c : \ichoice{\pot{l_i}{r_i} : S_i}_{i \in I})
\end{equation}
\begin{equation}\label{eqn:oplus_C_s_2}
q_2 \geq 0 \qquad c' : S_k \entailpot{q_2} \fwd{c}{c'} :: (c : S_k)
\end{equation}
\begin{equation}\label{eqn:oplus_C_s_3}
q_3 + r_k \geq q_k \qquad
\W \; (c : \ichoice{\pot{l_i}{r_i} : S_i}_{i \in I}) \entailpot{q_3} \ecase{c}{l_i}{Q_i}_{i \in I} :: (d : U)
\end{equation}
and in $\config_M'$ (premise of the typing rule due to inversion),
we get
\[
\W \; (c' : S_k) \entailpot{q_k} [c'/c]Q_k :: (d : U)
\]
Using the cost semantics rule $\oplus C_r$, we get
\[
\infer[\oplus C_r]
{\proc{d}{w+w', [c'/c]Q_k}}
{\msg{c}{w, \esendl{c}{l_k} \semi \fwd{c}{c'}} \qquad \proc{d}{w', \ecase{c}{l_i}{Q_i}_{i \in I}}}
\]
Again, $S_M' = q_k + w + w' \leq q_3 + r_k + w + w'
\leq q_3 + q_2 + r_k + w + w' \leq q_3 + q_1 + w + w'
\leq (q_1 + w) + (q_3 + w') = S_M$ (by Equations
\ref{eqn:oplus_C_s_1}, \ref{eqn:oplus_C_s_2} and
\ref{eqn:oplus_C_s_3}).

\item Case ($\with C_s$) : Analogous to $\lolli C_s$.

\item Case ($\with C_r$) : Analogous to $\lolli C_r$.

\item Case ($\tensor C_s$) : Analogous to $\oplus C_s$.

\item Case ($\tensor C_r$) : Analogous to $\oplus C_r$.

\item Case ($\lolli C_s$) : $\config_M = \proc{d}{w, \esendch{c}{e} \semi P}$.
Applying the rule $\lolli L$ on $\config_M$, we get
\[
q_1 \geq p + r + \mchan \qquad
\W \; (e : S) \; (c : S \lollipot{r} T) \entailpot{q_1} (\esendch{c}{e} \semi P) :: (d : U)
\]
Inverting the same rule on $\config_M'$, we get
\[
\W \; (c' : T) \entailpot{p} [c'/c]P :: (d : U)
\]
\[
q_1' \geq q_2 + r \qquad (e : S) \; (c : S \lollipot{r} T) \entailpot{q_1'} \esendch{c}{e} \semi \fwd{c'}{c} :: (c' : T)
\]
\[
q_2 \geq 0 \qquad c : T \entailpot{q_2} \fwd{c'}{c} :: (c' : T)
\]
From the cost semantics rule $\lolli C_s$, we get
\[
\infer[\lolli C_s]
{\proc{d}{w+\mchan, [c'/c]P} \qquad \msg{c'}{0, \esendch{c}{e} \semi \fwd{c'}{c}}}
{\proc{d}{w, \esendch{c}{e} \semi P}}
\]
Since we can choose arbitrary values for $q_1'$ and
$q_2$ satisfying the above inequalities, we set
$q_2 = 0$ and $q_1' = r$. Computing $S_M$
and $S_M'$, we get
\begin{eqnarray*}
S_M' & = & (p + w + \mchan) + (q_1'  + w') \\
& = & (p + r + \mchan + w) + (q_1' - r + w') \\
& \leq & q_1 + w + w') \\
& = & S_M
\end{eqnarray*}

\item Case ($\lolli C_r$) : $\config_M = \msg{c'}{w, \esendch{c}{e} \semi \fwd{c'}{c}} \; \proc{c}{w', \erecvch{c}{x} \semi Q_x}$.
Applying the rule $\lolli L$ on the message
in $\config_M$, we get
\[
q_1 \geq q_2 + r \qquad (e : S) \; (c : S \lollipot{r} T) \entailpot{q_1} \esendch{c}{e} \semi \fwd{c'}{c} :: (c' : T)
\]
\[
q_2 \geq 0 \qquad c : T \entailpot{q_2} \fwd{c'}{c} :: (c' : T)
\]
Applying the rule $\lolli R$ on the process
in $\config_M$, we get
\[
q_3 + r \geq p \qquad
\W \entailpot{q_3} (\erecvch{c}{x} \semi Q_x) :: (x : S \lollipot{r} T)
\]
Inverting the $\lolli R$ rule, we get for $\config_M'$,
\[
\W \; (e : S) \entailpot{p} [c'/c]Q_e :: (c' : T)
\]
From the cost semantics rule $\lolli C_r$, we get
\[
\infer[\lolli C_r]
{\proc{c}{w+w', [c'/c]Q_e}}
{\msg{c'}{w, \esendch{c}{e} \semi \fwd{c'}{c}} \qquad \proc{c}{w', \erecvch{c}{x} \semi Q_x}}
\]
Computing $S_M$ and $S_M'$, we get
\begin{eqnarray*}
S_M' & = & p + w + w' \\
& \leq & q_3 + r + w + w' \\
& \leq & q_3 + q_2 + r + w + w' \\
& \leq & q_3 + q_1 + w + w' \\
& = & (q_1 + w) + (q_3 + w') \\
& = & S_M
\end{eqnarray*}

\item Case ($\one C_s$) : $\config_M = \proc{c}{w, \eclose{c}}$.
Applying the $\one R$ rule on $\config_M$, we get
\[
q \geq r + \mcl \qquad
\cdot \entailpot{q} \eclose{c} :: (c : \pot{\one}{r})
\]
Inverting the same rule for $\config_M'$, we get
\[
q' \geq r \qquad \cdot \entailpot{q'} \eclose{c} :: (c : \pot{\one}{r})
\]
From the cost semantics rule $\one C_s$, we get
\[
\infer[\one C_s]
{\msg{c}{w+\mcl, \eclose{c}}}
{\proc{c}{w, \eclose{c}}}
\]
Setting $q' = r$ and computing $S_M$
and $S_M'$, we get
\begin{eqnarray*}
S_M' & = & q' + w + \mcl \\
& = & r + w + \mcl \\
& \leq & q + w \\
& = & S_M
\end{eqnarray*}

\item Case ($\one C_r$) : $\config_M = \msg{c}{w, \eclose{c}} \; \proc{d}{w', \ewait{c} \semi Q}$.
Applying the rule $\one R$ on the message
in $\config_M$, we get
\[
q_1 \geq r \qquad
\cdot \entailpot{q_1} \eclose{c} :: (c : \pot{\one}{r})
\]
Applying the $\one L$ rule on the process
in $\config_M$, we get
\[
q_2 + r \geq p \qquad
\W \; (c : \pot{\one}{r}) \entailpot{q_2} \ewait{c} \semi Q :: (d : U)
\]
Inverting the $\one L$ rule for $\config_M'$, we get
\[
\W \entailpot{p} Q :: (d : U)
\]
From the cost semantics rule $\one C_r$, we get
\[
\infer[\one C_r]
{\proc{d}{w + w', Q}}
{\msg{c}{w, \eclose{c}} \qquad \proc{d}{w', \ewait{c} \semi Q}}
\]
Computing $S_M$ and $S_M'$, we get
\begin{eqnarray*}
S_M' & = & p + w + w' + \mcl \\
& \leq & q_2 + r + w + w' \\
& \leq & q_2 + q_1 + w + w' \\
& \leq & (q_1 + w) + (q_2 + w') \\
& = & S_M
\end{eqnarray*}

\end{itemize}
Hence, in all of the above cases, $S_M' \leq S_M$
establishing that $S' \leq S$, thus showing that the
potential type system is sound w.r.t. the cost semantics.
\end{proof}

\section{More Examples}\label{sec:more}
Our type system is quite expressive and
can be used to derive bounds on many more
examples. In this section, we will derive bounds
on several list processes. We will start with
simple examples, such as the $nil$, $cons$
and $append$ processes. We will then derive
bounds on stacks and queues being implemented
using lists. Finally, we will conclude with some
higher order functions such as $map$ and $fold$.
For each of the following examples, we assume the
standard cost metric, where we count the number of
messages exchanged, i.e. $\mlab = \mchan = \mcl = 1$.
First, we consider the list protocol as a simple session
type.
\begin{tabbing}
$\lt{A} = \ichoice{$ \= $\m{cons} : A \tensor \lt{A},$ \\
\> $\m{nil} : \one}$
\end{tabbing}
The type prescribes that a process providing service
of type $\lt{A}$ will either send a label $\m{cons}$
followed by an element of type $A$ and recurse,
or will send a $\m{nil}$ label followed by a close
message and then terminate. On the client side
(i.e. a process that uses a channel of type $\lt{A}$
in its context), the opposite behavior is observed, i.e.
a client receives the messages that the provider sends
(sequence of $\m{cons}$ labels and elements terminated
by a $\m{nil}$ label and the close messsage).
The $append$ process is a SILL implementation
of the standard append function which appends two
lists.

\subsection{Basic Processes}
We present the implementations of $nil$, $cons$
and $append$ processes.

\subsection*{$nil$}
The $nil$ process is used to create an empty list.
Formally, a $nil$ process uses an empty context,
and provides an empty list along a channel $l : \lt{A}$.
Concretely, this means it sends a $\m{nil}$ label
followed by a $\m{close}$ message along $l$.
First, we introduce the resource-aware session type
for $\lt{A}$.
\begin{tabbing}
$\lt{A} = \ichoice{$ \= $\pot{\m{nil}}{0} : \pot{\one}{0},$ \\
\> $\pot{\m{cons}}{0} : A \tensorpot{0} \lt{A}}$
\end{tabbing}
This resource-aware type decorates each label
and type operator with $0$ potential, implying
that none of the messages carry any potential,
and the process potential needs to pay only for the
cost of sending the messages.
We present the implementation followed by the type
derivation for the $nil$ process.
\begin{tabbing}
$\cdot \; \entailpot{2} nil :: (l : \lt{A})$\\
$\quad$ \= $\procdefna{nil}{l} = $\\
\>\quad\=$\esendl{l}{\m{nil}} \semi$
\hspace{3em} \= $\%\quad \cdot \entailpot{1} l : \pot{\one}{0}$ \\
\>\>$\eclose{l}$
\> $\%\quad \cdot \entailpot{0} \cdot$
\end{tabbing}
The type of $nil$ process shows that the process
potential needed is $2$, which intuitively
agrees with our cost model. The process sends two
messages, each of them costing unit potential.
We explain the type derivation briefly. The initial
type of the process is $\cdot \entailpot{2} l : \lt{A}$.
Now, for $l$ to behave as an empty list, the $nil$
process needs to send the $\m{nil}$ label first.
As the $\lt{A}$ type prescribes, the $\m{nil}$ label
carries no potential, this send only costs $1$.
Updating the type of $l$ and the process potential,
we get $\cdot \entailpot{1} l : \pot{\one}{0}$.
 Finally, the $nil$ process needs to send the $\m{close}$
 message, which again costs $1$ as the type $\one$ carries no
 potential in the type definition of $\lt{A}$. Thus,
 our type system successully verifies that the $nil$
 process needs a potential of $2$, hence its resource
 usage is $2$.
 
 \subsection*{$cons$}
 Now, let's look at the $cons$ process.
\begin{tabbing}
$(x : A) \; (t : \pot{\lt{A}}{0}) \; \entailpot{2} cons :: (l : \pot{\lt{A}}{0})$\\
$\quad$ \= $\procdef{cons}{x \; t}{l} = $\\
\>\quad\=$\esendl{l}{\m{cons}} \semi$
\hspace{3em} \= $\%\quad (x : A) \; (t : \pot{\lt{A}}{0}) \entailpot{1} l : A \tensorpot{0} \pot{\lt{A}}{0}$ \\
\>\>$\esendch{l}{x} \semi$
\> $\%\quad (t : \pot{\lt{A}}{0}) \entailpot{0} l : \pot{\lt{A}}{0}$\\
\>\>$\fwd{l}{t}$
\end{tabbing}
We can do another annotated type for $cons$.
\begin{tabbing}
$\pot{\lt{A}}{1} = \ichoice{$ \= $\pot{\m{nil}}{0} : \one,$ \\
\> $\pot{\m{cons}}{1} : A \tensorpot{0} \pot{\lt{A}}{1}}$
\end{tabbing}
\begin{tabbing}
$(x : A) \; (t : \pot{\lt{A}}{1}) \; \entailpot{3} cons :: (l : \pot{\lt{A}}{1})$\\
$\quad$ \= $\procdef{cons}{x \; t}{l} = $\\
\>\quad\=$\esendl{l}{\m{cons}} \semi$
\hspace{3em} \= $\%\quad (x : A) \; (t : \pot{\lt{A}}{1}) \entailpot{1} l : A \tensorpot{0} \pot{\lt{A}}{1}$ \\
\>\>$\esendch{l}{x} \semi$
\> $\%\quad (t : \pot{\lt{A}}{1}) \entailpot{0} l : \pot{\lt{A}}{1}$\\
\>\>$\fwd{l}{t}$
\end{tabbing}

\subsection*{$append$}
Finally, let's try the $append$ process.
\begin{tabbing}
$\pot{\lt{A}}{2} = \ichoice{$ \= $\pot{\m{nil}}{0} : \one,$ \\
\> $\pot{\m{cons}}{2} : A \tensorpot{0} \pot{\lt{A}}{2}}$
\end{tabbing}
\begin{tabbing}
$(l_1 : \pot{\lt{A}}{2}) \; (l_2 : \pot{\lt{A}}{0}) \; \entailpot{0} append :: (l : \pot{\lt{A}}{0})$\\
$\quad$ \= $\procdef{append}{l_1 \; l_2}{l} = $\\
\>\quad\=$\casedef{l_1}$ \=$( \labdef{\m{cons}}$ \= $ \erecvch{l_1}{x}\semi$
\hspace{3em} \= $\%\quad 
(x : A) \; (l_1 : \pot{\lt{A}}{2}) \; (l_2 : \pot{\lt{A}}{0}) \entailpot{2} l : \pot{\lt{A}}{0}$ \\
\>\>\>\>$\esendl{l}{\m{cons}} \semi$
\> $\%\quad (x : A) \; (l_1 : \pot{\lt{A}}{2}) \; (l_2 : \pot{\lt{A}}{0}) \entailpot{1} l : \pot{\lt{A}}{0}$\\
\>\>\>\>$\esendch{l}{x} \semi$
\> $\%\quad (l_1 : \pot{\lt{A}}{2}) \; (l_2 : \pot{\lt{A}}{0}) \entailpot{0} l : \pot{\lt{A}}{0}$\\
\>\>\>\>$\procdef{append}{l_1 \; l_2}{l}$\\
\>\>\>$\mid \labdef{\m{nil}}$ \> $\ewait{l_1} \semi$
\> $\%\quad (l_1 : \one) \; (l_2 : \pot{\lt{A}}{0}) \entailpot{0} l : \pot{\lt{A}}{0}$\\
\>\>\>\> $\fwd{l}{l_2})$
\> $\%\quad (l_2 : \pot{\lt{A}}{0}) \entailpot{0} l : \pot{\lt{A}}{0}$\\
\end{tabbing}


\subsection{Stacks as Lists}
The stack interface introduced in Section~\ref{sec:case_study}
can be implemented using lists. First, we define the resource-aware
types.
\begin{tabbing}
$\pot{\lt{A}}{2} = \ichoice{$ \= $\pot{\m{nil}}{2} : \pot{\one}{0},$ \\
\> $\pot{\m{cons}}{2} : A \tensorpot{0} \pot{\lt{A}}{2}}$
\end{tabbing}

\begin{tabbing}
$\pot{\stack{A}}{4} = \echoice{$ \= $\pot{\m{ins}}{4} : A \lollipot{0} \pot{\stack{A}}{4},$ \\
\> $\pot{\m{del}}{0} : \ichoice{$ \= $\pot{\m{some}}{0} : A \tensorpot{0} \pot{\stack{A}}{4},$\\
\>\>$\pot{\m{none}}{0} : \pot{\one}{0}}}$
\end{tabbing}

We need several sub-processes for the implementation of the
stack using a list. We will implement and type them first.
\begin{tabbing}
$\cdot \; \entailpot{4} nil :: (l : \pot{\lt{A}}{2})$\\
$\quad$ \= $\procdefna{nil}{l} = $\\
\>\quad\=$\esendl{l}{\m{nil}} \semi$
\hspace{3em} \= $\%\quad \cdot \entailpot{1} l : \pot{\one}{0}$ \\
\>\>$\eclose{l}$
\> $\%\quad \cdot \entailpot{0} \cdot$
\end{tabbing}

\begin{tabbing}
$(x : A) \; (t : \pot{\lt{A}}{2}) \; \entailpot{4} cons :: (l : \pot{\lt{A}}{2})$\\
$\quad$ \= $\procdef{cons}{x \; t}{l} = $\\
\>\quad\=$\esendl{l}{\m{cons}} \semi$
\hspace{3em} \= $\%\quad (x : A) \; (t : \pot{\lt{A}}{2}) \entailpot{1} l : A \tensorpot{0} \pot{\lt{A}}{2}$ \\
\>\>$\esendch{l}{x} \semi$
\> $\%\quad (t : \pot{\lt{A}}{2}) \entailpot{0} l : \pot{\lt{A}}{2}$\\
\>\>$\fwd{l}{t}$
\end{tabbing}

Finally, we can implement the stack interface using two processes,
the first is $stack\_new$, which creates an empty list and uses it as
an empty stack.
\begin{tabbing}
$\cdot \entailpot{4} stack\_new :: s : (\pot{\stack{A}}{4})$\\
$\quad$ \= $\procdefna{stack\_new}{s} = $\\
\>\quad \= $\procdefna{nil}{e} \semi$
\hspace{3em} \= $(e : \pot{\lt{A}}{2}) \entailpot{0} (s : \pot{\stack{A}}{4})$\\
\>\> $\procdef{stack}{e}{s}$
\end{tabbing}

The main process is called $stack$. It uses a list in its context and
provides service along $s$ which behaves as a stack.
\begin{tabbing}
$l : \pot{\lt{A}}{2} \entailpot{0} stack :: s : (\pot{\stack{A}}{4})$ \\
\quad \= $\procdef{stack}{l}{s}$ \\
\> \quad \= $\casedef{s}$\\
\>\> \quad \= $(\labdef{\m{ins}}$ \= $\erecvch{s}{x} \semi$
\hspace{3em} \= $\% \quad (x : A)(l : \pot{\lt{A}}{2}) \entailpot{4} s : \pot{\stack{A}}{2}$\\
\>\>\>\> $\procdef{cons}{x \; t}{l'}$
\>$\% \quad l' : \pot{\lt{A}}{2} \entailpot{0} s : \pot{\stack{A}}{4}$\\
\>\>\>\> $\procdef{stack}{l}{s}$\\
\>\>\> $\labdef{\m{del}}$
\> $\casedef{l}$\\
\>\>\>\>\quad \= $(\labdef{\m{cons}}$ \= $\erecvch{l}{x} \semi$
\hspace{2em} \= $\% \quad (x : A)(l : \pot{\lt{A}}{2}) \entailpot{2} s : \ichoice{\pot{\m{some}}{0} : A \tensorpot{0} \pot{\stack{A}}{4}, \pot{\m{none}}{0} : \pot{\one}{0}}$\\
\>\>\>\>\>\> $\esendl{s}{\m{some}} \semi$
\> $\% \quad (x : A)(l : \pot{\lt{A}}{2}) \entailpot{1} s : A \tensorpot{0} \pot{\stack{A}}{4}$\\
\>\>\>\>\>\> $\esendch{s}{x} \semi$
\> $\% \quad (l : \pot{\lt{A}}{2}) \entailpot{0} s : \pot{\stack{A}}{4}$\\
\>\>\>\>\>\> $\procdef{stack}{l}{s}$\\
\>\>\>\>\> $\labdef{\m{nil}}$
\> $\esendl{s}{\m{none}}$
\> $\% \quad (l : \pot{\one}{0}) \entailpot{1} s : \pot{\one}{0}$\\
\>\>\>\>\>\> $\ewait{l} \semi$
\> $\% \quad \cdot \entailpot{1} s : \pot{\one}{0}$\\
\>\>\>\>\>\> $\eclose{s}))$
\end{tabbing}

\subsection{Queues as 2 Lists}\label{subsec:app_queue}
A queue can be implemented using 2 lists. Insertion
in such a queue has a constant amortized cost. Since our
type system supports amortized analysis, we can derive
a constant resource bound for such an implementation.
The resource-aware types we will be using are as follows.
\begin{tabbing}
$\pot{\lt{A}}{(2,2)} = \ichoice{$ \= $\pot{\m{nil}}{2} : \pot{\one}{0},$ \\
\> $\pot{\m{cons}}{2} : A \tensorpot{0} \pot{\lt{A}}{(2,2)}}$
\end{tabbing}

\begin{tabbing}
$\pot{\lt{A}}{4} = \ichoice{$ \= $\pot{\m{nil}}{0} : \pot{\one}{0},$ \\
\> $\pot{\m{cons}}{4} : A \tensorpot{0} \pot{\lt{A}}{4}}$
\end{tabbing}

\begin{tabbing}
$\pot{\queue{A}}{(6,2)} = \echoice{$ \= $\pot{\m{enq}}{6} : A \lollipot{0} \pot{\queue{A}}{(6,2)},$ \\
\> $\pot{\m{deq}}{2} : \ichoice{$ \= $\pot{\m{some}}{0} : A \tensorpot{0} \pot{\queue{A}}{(6,2)},$\\
\>\>$\pot{\m{none}}{0} : \pot{\one}{0}}}$
\end{tabbing}

Again, we use the $nil$ and $cons$ sub-processes with a
different type.
\begin{tabbing}
$\cdot \; \entailpot{2} nil :: (l : \pot{\lt{A}}{4})$\\
$\quad$ \= $\procdefna{nil}{l} = $\\
\>\quad\=$\esendl{l}{\m{nil}} \semi$
\hspace{3em} \= $\%\quad \cdot \entailpot{1} l : \pot{\one}{0}$ \\
\>\>$\eclose{l}$
\> $\%\quad \cdot \entailpot{0} \cdot$
\end{tabbing}

\begin{tabbing}
$(x : A) \; (t : \pot{\lt{A}}{4}) \; \entailpot{6} cons :: (l : \pot{\lt{A}}{4})$\\
$\quad$ \= $\procdef{cons}{x \; t}{l} = $\\
\>\quad\=$\esendl{l}{\m{cons}} \semi$
\hspace{3em} \= $\%\quad (x : A) \; (t : \pot{\lt{A}}{4}) \entailpot{1} l : A \tensorpot{0} \pot{\lt{A}}{4}$ \\
\>\>$\esendch{l}{x} \semi$
\> $\%\quad (t : \pot{\lt{A}}{4}) \entailpot{0} l : \pot{\lt{A}}{4}$\\
\>\>$\fwd{l}{t}$
\end{tabbing}

The main process $queue2$ acts as a client for 2 lists, and provides
service along a queue interface. We provide the implementation and
its type derivation below.
{\footnotesize
\begin{tabbing}
$(in : \pot{\lt{A}}{4}) \; (out : \pot{\lt{A}}{(2,2)}) \entailpot{0} queue2 :: s : (\pot{\queue{A}}{(6,2)})$ \\
\quad \= $\procdef{queue2}{l}{s}$ \\
\> \quad \= $\casedef{s}$\\
\>\> \quad \= $(\labdef{\m{enq}}$ \= $\erecvch{s}{x} \semi$
\hspace{5em} \= $\% \quad (x : A)(in : \pot{\lt{A}}{4}) \; (out : \pot{\lt{A}}{(2,2)}) \entailpot{6} s : \pot{\queue{A}}{(6,2)}$\\
\>\>\>\> $\procdef{cons}{x \; in}{in'}$
\>$\% \quad (in' : \pot{\lt{A}}{4}) \; (out : \pot{\lt{A}}{(2,2)}) \entailpot{0} s : \pot{\queue{A}}{(6,2)}$\\
\>\>\>\> $\procdef{queue2}{in \; out}{s}$\\
\>\>\> $\labdef{\m{deq}}$
\> $\casedef{out}$\\
\>\>\>\>\quad \= $(\labdef{\m{cons}}$ \= $\erecvch{out}{x} \semi$
\hspace{0.2em} \= $\% \quad (x : A) \; (in : \pot{\lt{A}}{4}) \; (out : \pot{\lt{A}}{(2,2)}) \entailpot{2}$ \\
\>\>\>\>\>\>\> \% \hspace{2em} $ s : \ichoice{\pot{\m{some}}{0} : A \tensorpot{0} \pot{\queue{A}}{(6,2)}, \pot{\m{none}}{0} : \pot{\one}{0}}$\\
\>\>\>\>\>\> $\esendl{s}{\m{some}} \semi$
\> $\% \quad (x : A) \; (in : \pot{\lt{A}}{2}) \; (out : \pot{\lt{A}}{0}) \entailpot{1} s : A \tensorpot{0} \pot{\queue{A}}{4}$\\
\>\>\>\>\>\> $\esendch{s}{x} \semi$
\> $\% \quad (in : \pot{\lt{A}}{2}) \; (out : \pot{\lt{A}}{0}) \entailpot{0} s : \pot{\queue{A}}{4}$\\
\>\>\>\>\>\> $\procdef{queue2}{in \; out}{s}$\\
\>\>\>\>\> $\labdef{\m{nil}}$
\> $\ewait{out} \semi$
\> $\% \quad (in : \pot{\lt{A}}{2}) \entailpot{4} s : \pot{\queue{A}}{4}$\\
\>\>\>\>\>\> $\procdef{rev}{in}{out'} \semi$
\> $\% \quad (out' : \pot{\lt{A}}{0}) \entailpot{2} s : \pot{\queue{A}}{4}$\\
\>\>\>\>\>\> $\casedef{out'}$\\
\>\>\>\>\>\>\quad \= $(\labdef{\m{cons}}$ \= $\erecvch{out'}{x} \semi$
\hspace{0em} \= $\% \quad (x : A) \; (out' : \pot{\lt{A}}{0}) \entailpot{2}$\\
\>\>\>\>\>\>\>\>\> $\% \quad \quad \quad s : \ichoice{\pot{\m{some}}{0} : A \tensorpot{0} \pot{\queue{A}}{4}, \pot{\m{none}}{0} : \pot{\one}{0}}$\\
\>\>\>\>\>\>\>\> $\esendl{s}{\m{some}} \semi$
\> $\% \quad (x : A) \; (out' : \pot{\lt{A}}{0}) \entailpot{1} s : A \tensorpot{0} \pot{\queue{A}}{4}$\\
\>\>\>\>\>\>\>\> $\esendch{s}{x} \semi$
\> $\% \quad (out' : \pot{\lt{A}}{0}) \entailpot{0} s : \pot{\queue{A}}{4}$\\
\>\>\>\>\>\>\>\> $\procdefna{nil}{in'} \semi$
\> $\% \quad (in' : \pot{\lt{A}}{2}) \; (out' : \pot{\lt{A}}{0}) \entailpot{0} s : \pot{\queue{A}}{4}$\\
\>\>\>\>\>\>\>\> $\procdef{queue2}{in' \; out'}{s}$\\
\>\>\>\>\>\>\> $\labdef{\m{nil}}$
\> $\ewait{out'} \semi$
\> $\% \quad \cdot \entailpot{0} s : \ichoice{\pot{\m{some}}{0} : A \tensorpot{0} \pot{\queue{A}}{4}, \pot{\m{none}}{0} : \pot{\one}{0}}$\\
\>\>\>\>\>\>\>\> $\esendl{s}{\m{none}} \semi$
\> $\% \quad \cdot \entailpot{0} s : \pot{\one}{0}$\\
\>\>\>\>\>\>\>\> $\eclose{s} )))$
\end{tabbing}
}

\subsection{Higher Order Functions}
Let's consider some higher order functions. First, let's consider the
$map$ function.

\begin{tabbing}
$\pot{\mapper{AB}}{2} = \echoice{$ \= $\pot{\m{next}}{0} : A \lollipot{0} B \tensorpot{2} \pot{\mapper{AB}}{2},$ \\
\> $\pot{\m{done}}{0} : \pot{\one}{2}}$
\end{tabbing}

\begin{tabbing}
$\pot{\lt{A}}{2} = \ichoice{$ \= $\pot{\m{nil}}{1} : \pot{\one}{0},$ \\
\> $\pot{\m{cons}}{2} : A \tensorpot{0} \pot{\lt{A}}{2}}$
\end{tabbing}

\begin{tabbing}
$\pot{\lt{A}}{0} = \ichoice{$ \= $\pot{\m{nil}}{0} : \pot{\one}{0},$ \\
\> $\pot{\m{cons}}{0} : A \tensorpot{0} \pot{\lt{A}}{0}}$
\end{tabbing}
Now, let's consider the implementation of the $map$ function.
\begin{tabbing}
$(l : \pot{\lt{A}}{2}) \; (m : \pot{\mapper{AB}}{2}) \entailpot{0} map :: (k : \pot{\lt{A}}{0})$\\
\quad \= $\procdef{map}{l \; m}{k} =$\\
\> \quad \= $\casedef{l}$\\
\>\> \quad \= $(\labdef{\m{cons}}$ \= $\erecvch{l}{x} \semi$
\hspace{3em} \= $\% \quad (x : A) \; (l : \pot{\lt{A}}{2}) \; (m : \pot{\mapper{AB}}{2})
\entailpot{2} (k : \pot{\lt{A}}{0})$\\
\>\>\>\> $\esendl{m}{\m{next}} \semi$
\> $\% \quad (x : A) \; (l : \pot{\lt{A}}{2}) \; (m : A \lollipot{0} B \tensorpot{2} \pot{\mapper{AB}}{2})
\entailpot{1} (k : \pot{\lt{A}}{0})$\\
\>\>\>\> $\esendch{m}{x} \semi$
\> $\% \quad (l : \pot{\lt{A}}{2}) \; (m : B \tensorpot{2} \pot{\mapper{AB}}{2}) \entailpot{0} (k : \pot{\lt{A}}{0})$\\
\>\>\>\> $\erecvch{m}{y} \semi$
\> $\% \quad (l : \pot{\lt{A}}{2}) \; (y : B) \; (m : \pot{\mapper{AB}}{2}) \entailpot{2} (k : \pot{\lt{A}}{0})$\\
\>\>\>\> $\esendl{k}{\m{cons}} \semi$
\> $\% \quad (l : \pot{\lt{A}}{2}) \; (y : B) \; (m : \pot{\mapper{AB}}{2}) \entailpot{1} (k : A \tensorpot{0} \pot{\lt{A}}{0})$\\
\>\>\>\> $\esendch{k}{y} \semi$
\> $\% \quad (l : \pot{\lt{A}}{2}) \; (m : \pot{\mapper{AB}}{2}) \entailpot{0} (k : \pot{\lt{A}}{0})$\\
\>\>\>\> $\procdef{map}{l \; m}{k}$\\
\>\>\> $\labdef{\m{nil}}$
\> $\ewait{l} \semi$
\> $\% \quad (m : \pot{\mapper{AB}}{2}) \entailpot{1} (k : \pot{\lt{A}}{0})$\\
\>\>\>\> $\esendl{m}{\m{done}} \semi$
\> $\% \quad (m : \pot{\one}{2}) \entailpot{0} (k : \pot{\lt{A}}{0})$\\
\>\>\>\> $\ewait{m} \semi$
\> $\% \quad \cdot \entailpot{2} (k : \pot{\lt{A}}{0})$\\
\>\>\>\> $\esendl{k}{\m{nil}} \semi$
\> $\% \quad \cdot \entailpot{1} (k : \pot{\one}{0})$\\
\>\>\>\> $\eclose{k} )$
\end{tabbing}
Now, let's consider the $fold$ function.
\begin{tabbing}
$\pot{\fdr{AB}}{0} = \echoice{$ \= $\pot{\m{next}}{0} : A \lollipot{0} B \lollipot{0} B \tensorpot{0} \pot{\fdr{AB}}{0},$ \\
\> $\pot{\m{done}}{0} : \pot{\one}{0}}$
\end{tabbing}

\begin{tabbing}
$\pot{\lt{A}}{3} = \ichoice{$ \= $\pot{\m{nil}}{1} : \pot{\one}{0},$ \\
\> $\pot{\m{cons}}{3} : A \tensorpot{0} \pot{\lt{A}}{2}}$
\end{tabbing}
Now, let's consider the implementation of the $fold$ function.
\begin{tabbing}
$(l : \pot{\lt{A}}{3}) \; (m : \pot{\fdr{AB}}{0}) \; (b : B) \entailpot{0} fold :: (r : B)$\\
\quad \= $\procdef{fold}{l \; m \; b}{r} =$\\
\> \quad \= $\casedef{l}$\\
\>\> \quad \= $(\labdef{\m{cons}}$ \= $\erecvch{l}{x} \semi$
\hspace{3em} \= $\% \quad (x : A) \; (l : \pot{\lt{A}}{2}) \; (m : \pot{\fdr{AB}}{0}) \; (b : B)
\entailpot{3} (r : B)$\\
\>\>\>\> $\esendl{m}{\m{next}} \semi$
\> $\% \quad (x : A) \; (l : \pot{\lt{A}}{3}) \; (m : A \lollipot{0} B \lollipot{0} B \tensorpot{2} \pot{\fdr{AB}}{0}) \; (b : B)
\entailpot{2} (r : B)$\\
\>\>\>\> $\esendch{m}{x} \semi$
\> $\% \quad (l : \pot{\lt{A}}{3}) \; (m : B \lollipot{0} B \tensorpot{2} \pot{\fdr{AB}}{0}) \; (b : B)\entailpot{1} (r : B)$\\
\>\>\>\> $\esendch{m}{b} \semi$
\> $\% \quad (l : \pot{\lt{A}}{3}) \; (m : B \tensorpot{2} \pot{\fdr{AB}}{0})\entailpot{0} (r : B)$\\
\>\>\>\> $\erecvch{m}{y} \semi$
\> $\% \quad (l : \pot{\lt{A}}{2}) \; (y : B) \; (m : \pot{\fdr{AB}}{0}) \entailpot{2} (r : B)$\\
\>\>\>\> $\procdef{map}{l \; m}{k}$\\
\>\>\> $\labdef{\m{nil}}$
\> $\ewait{l} \semi$
\> $\% \quad (m : \pot{\fdr{AB}}{0}) \; (b : B) \entailpot{1} (r : B)$\\
\>\>\>\> $\esendl{m}{\m{done}} \semi$
\> $\% \quad (m : \pot{\one}{0}) \; (b : B) \entailpot{0} (r : B)$\\
\>\>\>\> $\ewait{m} \semi$
\> $\% \quad (b : B) \entailpot{0} (r : B)$\\
\>\>\>\> $\fwd{r}{b})$
\end{tabbing}


\end{document}